\documentclass[final]{lipics-v2021}

\usepackage{wrapfig}
\usepackage{graphicx} 
\usepackage{amssymb,amsthm,amsmath}
\usepackage{algorithm}
\usepackage[noend]{algpseudocode}
\usepackage{comment}
\usepackage{tikz}
\usepackage{subcaption}
\usetikzlibrary{arrows.meta, decorations.pathmorphing}

\newif\iffull\fullfalse

\newcommand{\en}{{\tt n}}
\newcommand{\ex}{{\tt x}}
\newcommand{\EN}{N}
\newcommand{\EX}{X}
\newcommand{\PP}{{\cal P}}
\newcommand{\SP}{{\cal SP}}

\newcommand{\MS}{{\bf MS}}

\newcommand{\NP}{\emph{NP}}
\newcommand{\edge}{\!\rightarrow\!}
\usepackage{tikz-cd}
\usetikzlibrary{decorations.pathmorphing}
\usepackage[customcolors]{hf-tikz}
\hfsetfillcolor{gray!20} 
\hfsetbordercolor{white}
\newcommand{\pat}{\rightsquigarrow}
\newcommand{\W}{\cal W}

\title{An Incremental Algorithm for Checking the Possibility of Braess Paradox in Dynamic Nets}
\author{Dario {Fiorenza}}{Sapienza, University of Rome, Italy}{fiorenza@di.uniroma1.it}{https://orcid.org/0000-0001-8216-9886}{}
\author{Daniele {Gorla}}{Sapienza, University of Rome, Italy}{gorla@di.uniroma1.it}{https://orcid.org/0000-0001-8859-9844}{}
\author{Ivano {Salvo}}{Sapienza, University of Rome, Italy}{salvo@di.uniroma1.it}{https://orcid.org/0000-0003-3111-701X}{}
\authorrunning{D. Fiorenza, D. Gorla, I. Salvo}
\Copyright{Dario Fiorenza, Daniele Gorla, Ivano Salvo}


\date{April 2026}

\begin{document}
\nolinenumbers
\maketitle
\sloppy

\begin{abstract}
Braess paradox is a well-known phenomenon that originates when latency at Wardrop equilibrium in traffic networks decreases because of removing edges.
The possibility of having the paradox was called {\em vulnerability} by Roughgarden in 2006 and was characterized later on by graph-theoretical notions, both for undirected and for directed graphs.
In this paper we provide an algorithm for the incremental case of checking vulnerability for dynamically evolving graphs.
The crucial idea to keep the amortized cost linear for every edge addition is that we do not need to run the vulnerability algorithm on the whole graph, but only on a well-identified subgraph, determined by the edge that we are adding.
Overall, to add $m$ edges, we pay a cost of $O(m^2)$; this aligns with the $O(m^2)$ cost of the state-of-the-art static algorithm for vulnerability.
\end{abstract}


\section{Introduction}
{\em Traffic networks} \cite{braessFormal,BI97} provide a model for studying selfish routing, where   
non-cooperative agents travel from a source node $s$ to 
a destination node $t$. Since the latency experienced by an agent 
while traveling along a path depends on network congestion, 
traffic in a network is modeled as a non-cooperative game, 
which stabilizes to equilibrium when all agents cannot improve their latency by 
choosing a different path.
This notion of equilibrium was defined by Wardrop \cite{War52} in the
context of transport analysis.
In the selfish routing model, edges of a directed graph are labeled by a function 
specifying the latency that agents experience traveling on an edge 
in terms of the flow passing through it. 

As a representative example, consider the so-called {\em Wheatstone network} in Fig.~\ref{fig:Wheat}(1).
The edges $u\edge v$, $s\edge v$ and $u\edge t$ are  labeled by the constant functions 0 and 1: 
independently of how much traffic travels along them,
the passage from $u$ to $v$ is instantaneous, whereas it creates a delay of 1 in the other two cases. 
Finally, the delay on edges $s\edge u$ and $v\edge t$ 
is linear in the amount of traffic traveling along them.
The (non-optimal) flow of value $1$ where all selfish agents choose the path $s\edge u\edge v\edge t$ is at Wardrop equilibrium, 
and the overall delay experienced is $1+0+1=2$.
In contrast, the (optimal) flow in which agents distribute half along the path $s\edge u\edge t$ and half 
along $s\edge v\edge t$ is {\em not} at equilibrium, but the delay experienced is smaller (viz., 3/2). 
This is indeed the delay obtained with a flow of 
value 1 at Wardrop equilibrium in the subnet of the Wheatstone network depicted in Fig.~\ref{fig:Wheat}(2).
{\em Braess paradox} \cite{braessFormal,braessOriginal} 
originates when latency at Wardrop equilibrium decreases because of removing edges:
the Wheatstone network is a minimal example of this counterintuitive phenomenon. 

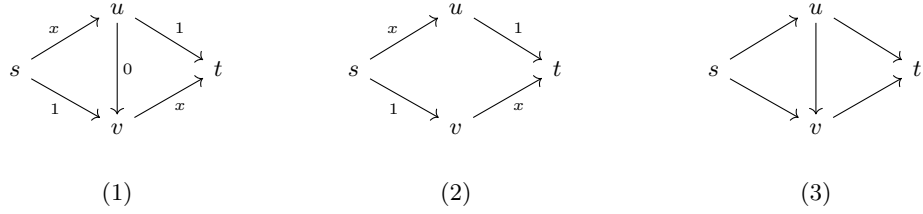
\begin{figure}[t]
	\begin{tabular}{ccc}
		\begin{minipage}{0.29\textwidth}
			\center
			\begin{tikzcd}[column sep=.8cm,row sep=.4cm]
			& u \ar[rd, "1"]\ar[dd,"0"]\\
			s \ar[ur,"x"] \ar[dr,"1" '] && t\\
			& v \ar[ru,"x" ']\\
			&(1)
			\\
			\end{tikzcd}
		\end{minipage}
		&
		\begin{minipage}{0.29\textwidth}
			\center
			\begin{tikzcd}[column sep=.8cm,row sep=.4cm]
			& u \ar[rd,"1"]\\
			s \ar[ur,"x"] \ar[dr,"1" '] && t\\
			& v \ar[ru,"x" ']\\
			&(2)
			\\
			\end{tikzcd}
		\end{minipage}
		&
		\begin{minipage}{0.33\textwidth}
			\center
			\begin{tikzcd}[column sep=.8cm,row sep=.4cm]
			& u \ar[dd]\ar[dr] & \\
			s \ar[ru]\ar[rd] & & t\\ 
			& v \ar[ru]\\
			&(3)
			\\
			\end{tikzcd}
		\end{minipage}
		\vspace{-3mm}
	\end{tabular}
	\caption{(1) The Wheatstone network; (2) Its optimal subgraph; (3) The graph $\W$.}
	\label{fig:Wheat}
	\vspace{-3mm}
\end{figure}

Braess paradox has been studied for decades. The results that are related with ours
start with \cite{Rough06}, where it is shown that, given a multigraph, a latency
function on its edges and a total amount of flow, it is \NP-hard to prove
whether or not the resulting network suffers from the Braess paradox. 
A change of perspective suggested in \cite{Rough06} is to study the
Braess paradox from a graph-theoretical point of view: in particular, a problem left open in \cite{Rough06} is the characterization of {\em vulnerable} $st$-graphs, which are those that admit instances (i.e., assignments of latency functions to edges) that generate the paradox. 
A characterization of vulnerable undirected multigraphs is presented in~\cite{Milch06}, where
it is proved that an undirected graph is vulnerable if and only if it is not series-parallel \cite{RS42}. This characterization only holds for graphs where all nodes and edges lie on at least one 
simple $st$-path. In \cite{ChenEtal15}, the same characterization is proved for directed multigraphs that satisfy the last condition, called {\em irredundancy} therein. 

However, while checking (ir)redundancy and finding the Maximum Irredundant Subgraph (MIS) can be efficiently done for undirected graphs \cite{CinesiFull}, the same does not hold for the directed ones: in \cite{CGS19} it is proved that recognizing if an edge is redundant (i.e., if it does not belong to any simple $st$-path) in a digraph is an NP-Hard problem.
Moreover, in \cite{CGS18} the authors prove that vulnerability coincides with containing an irredundant subgraph homeomorphic to the graph $\W$ of Fig.~\ref{fig:Wheat}(3)
and in \cite{CGS19} they provide a $O(nm^2)$ algorithm to state if a given graph is vulnerable or not; in \cite{MS23} a faster algorithm (that we call \MS\ in the following from the names of authors) is given which runs in $O(m^2)$.
Both these algorithms have the peculiar property that, if the input graph is not vulnerable, they also return its MIS.

Hence, the class of non-vulnerable graphs is strongly related to the very well-known and studied class of series-parallel graphs. Indeed, the two classes coincide in the undirected case; in the directed one, a graph is non-vulnerable if and only if its MIS is series-parallel.
Once these characterizations have been proved and the associated polynomial-time algorithms have been devised, the next challenging algorithmic question is how to handle dynamically evolving graphs in this context. Indeed, inter-node connections can change during time, either by appearing or by disappearing; this may consequently change the possibility of having the Braess Paradox. 
We stress that in this paper the term ‘‘dynamic'' refers exclusively to the evolution of the network topology (more precisely, to the possibility of adding edges) and does not refer to flows that change over time, as in the case, e.g., of \cite{MLS13,9921998}. Indeed, vulnerability is a notion that does not depend at all on the flows involved and its aim is precisely to characterise those topologies that, for whatever flow and demand, are immune to the Braess's paradox.

Having an algorithm that efficiently answers the vulnerability question without considering each time the whole graph is crucial, mostly in the dynamic case, where edges are added and removed sequentially without any information on future changes.
Dynamic algorithms are nowadays a mature research area \cite{BDGA23,chen2025recent,DynBook,HenFulDyn22}; they face, in a dynamic setting,  many important graph problems like, e.g.,
 the dynamic tree structure introduced  \cite{TarDyTr}, the graph connectivity \cite{JR87} and the minimum spanning tree \cite{JTDynProb01}.

In this paper we provide an algorithm for vulnerability in the incremental case, where edges are added one at a time. 
We start with a non-vulnerable graph $G$, since vulnerable graphs remain so after any edge addition. $G$ being not vulnerable, we can have its MIS by running \MS, i.e. the vulnerability algorithm in \cite{MS23}. 
When adding a new edge $e$, this can be redundant or irredundant in the new graph; moreover, its addition can make irredundant edges that were redundant in $G$. This can potentially open new paths which can create the homeomorphic copy of $\W$ inside the MIS of $G \cup \{e\}$, thus making the latter graph vulnerable.
So, the key aspect of our work is to find a linear time way to ascertain this fact and to calculate the new MIS, in the case in which non-vulnerability is preserved after the edge addition. To get this, the crucial observation is that we do not need to run again \MS\ on the whole resulting graph but only on the  subgraph formed by the edges that become irredundant thanks to the edge addition (if any).
Note that this subgraph can be arbitrary big; so, in general we cannot beat the cost of the static algorithm in one single iteration. However, we can prove that the amortized cost of every addition is $O(m)$, since the cost of the algorithm depends on the irredundant edges and, once an edge enters into the MIS, it will not be considered anymore in successive runs of \MS. 

To sum up, the main contribution of our paper is an algorithm that efficiently deals with edge additions by restricting as much as possible the part of the graph on which \MS\ is invoked as a sub-procedure and by avoiding analyzing that part more than once.
To obtain this, we start by refining the cost analysis of \MS, by showing that it is $O(km)$, where $k$ is the number of irredundant edges of the graph (the worst case clearly remains $O(m^2)$). 

The rest of the paper is organized as follows. In Section 2 we provide some preliminary background and definitions. In Section 3 we characterize the kinds of edge addition: impactless (the new edge is redundant in the resulting graph), impactful (the new edge is irredundant in the resulting graph), or destructive (the new edge makes the resulting graph vulnerable). Section 4 presents the algorithms for linearly checking these conditions, the incremental algorithm and its amortized analysis. Section 5 concludes the paper by drawing orthogonal and future research directions.
Most proofs and some auxiliary results are contained in the Appendix.

\section{Preliminaries}

\label{sec:vuln}

%
A \emph{directed multigraph} (or  \emph{multi-digraph}) $G = (V,E)$ consists of a
set $V$ of \emph{vertices} (or \emph{nodes}) and a set $E$ of \emph{edges}. 
Every edge relates a pair of vertices; if $e$ relates $(u,v)$, we say that $e$ is
an \emph{output} of $u$ and an \emph{input} of $v$.
%
When the name of an edge is not relevant but only its extremes are,
we denote an edge with the pair it relates.
 
A \emph{path} $P$ is a sequence $u_1 e_1 u_2 \ldots u_{n-1} e_{n-1} u_n$ (for $n \geq 1$) of nodes and edges 
such that $e_i$ relates $(u_i,u_{i+1})$, for all $i<n$;
$u_1$ and $u_n$ are called {\em extremes} and $u_2,\ldots,u_{n-1}$ are called {\em internal nodes}.
When only the extremes of the path $P$ are relevant, we write
$u_1 \stackrel P \pat u_n$ (or simply $u_1 \pat u_n$ if we ignore the name given to the path).
If no node appears more than once in $P$, we say that $P$ is {\em simple} (or {\em acyclic}).
We fix two different nodes: a \emph{source}, written $s$, and a \emph{target}, written $t$.
The resulting triple $(G,s,t)$ is called {\em net}; when $s$ and $t$ are understood, we sometimes
simply write it as $G$ (i.e., without specifying the source and target) and use the term ‘graph' instead of ‘net'.
An \emph{st-path} is a 
path from $s$ to $t$. The set of $st$-paths in $G$ is denoted by $\PP(G)$; 
the set of 
simple $st$-paths is denoted by $\SP(G)$. 
We say that a net is $st$-{\em connected} if every node is touched by at least 
one $st$-path.
%
%

\label{sec:IRRED}

The formal definitions of Braess Paradox and vulnerable nets are deferred to Appendix~\ref{app:braess}, 
since they are not fundamental to appreciate our results. Indeed, for our aims it suffices to know that a net is vulnerable if and only if one of the following two characterizations hold. 

The first characterisation of vulnerable nets holds for {\em irredundant} nets,
defined as those nets containing only irredundant vertices and edges.
As defined in \cite{ChenEtal15} (and derived from \cite{Milch06}), a vertex or an edge is {\em irredundant}
if it appears in a simple $st$-path, and {\em redundant} otherwise. 
The characterisation given by \cite{CinesiFull} states that an irredundant net is vulnerable if and only if it is series-parallel
(this can be checked in linear time \cite{VTL82}).

\begin{definition}[Two-Terminal Series-Parallel Graphs (TTSP) \cite{Duf65}]
A graph $G=(V,E)$, with two distinguished vertices $s$ and $t$, is a {\em TTSP-graph} if it can be turned into the single edge $(s,t)$ by a sequence of the following operations:
\begin{itemize}
\item Replace a pair of parallel edges with a single edge that connects their common endpoints;
\item Replace a pair of edges $(u,w)$ and $(w,v)$ with the edge $(u,v)$, whenever $w$ has degree 2 and it is different from $s$ and $t$.
\end{itemize}
\end{definition}

The second characterisation of vulnerability holds for generic (i.e., not necessarily irredundant) graphs. Indeed, verifying if a net is irredundant is an \NP-complete problem, as shown in \cite{CGS19}.
Hence, in \cite{CGS18} the authors give a characterisation of vulnerability based on the embedding of $\W$ within the given graph. 

\begin{definition}[\cite{LR80}] 
	A {\em subgraph homeomorphism} from $H$ to $G$
	is a pair of injective mappings $(\phi,\psi)$, from $V_H$ to $V_G$ and
	from $E_H$ to $\SP(G)$ respectively, 
	such that $\psi(e) = \phi(x) \leadsto \phi(y)$, for every $e \in  E_H$ that relates $(x,y)$.
	A homeomorphism is called {\em node-disjoint} if all paths in the image of $\psi$ are pairwise 
	node-disjoint, up-to their extremes.
\end{definition}

\begin{definition}
	\label{def:stEmbedding}
	An {\em st-embedding} of the Wheatstone graph $\W$ from Fig.~\ref{fig:Wheat}(3) into a net $(G,s',t')$
	is a node-disjoint subgraph homeomorphism $(\phi,\psi)$ from $\W$ to $G$ such that
	there are (possibly empty) node-disjoint simple paths from $s'$ to $\phi(s)$ and from $\phi(t)$ to $t'$ that
	are pairwise node-disjoint up-to their extremes with all paths in the image of $\psi$.
\end{definition}

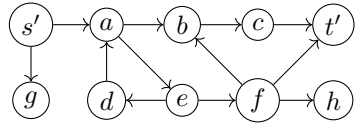
\begin{figure}[t]
    \centering

\begin{tikzpicture}[every node/.style={circle,draw, minimum size=4mm, inner sep=2pt}]
\node (s) at (0,0) {$s'$};
\node (a) at (1,0) {$a$};
\node (b) at (2,0) {$b$};
\node (c) at (3,0) {$c$};
\node (t) at (4,0) {$t'$};
\node (d) at (1,-1) {$d$};
\node (e) at (2,-1) {$e$};
\node (f) at (3,-1) {$f$};
\node (g) at (0,-1) {$g$};
\node (h) at (4,-1) {$h$};

\draw [->] (s) -- (a);
\draw [->] (a) -- (b);
\draw [->] (b) -- (c);
\draw [->] (c) -- (t);
\draw [->] (a) to (e);
\draw [->] (e) to (f);
\draw [->] (f) to (b);
\draw [->] (f) to (t);
\draw [->] (e) to (d);
\draw [->] (d) to (a);
\draw [->] (s) to (g);
\draw [->] (f) to (h);
\end{tikzpicture}
    \caption{A vulnerable graph containing three redundant nodes and four redundant edges}
    \label{fig:examplegraph}
\end{figure}

To better clarify these concepts, consider the graph in Fig.~\ref{fig:examplegraph}.
There, the edges relating $(e,d)$ and $(d,a)$ are redundant, since they can only appear in cyclic (i.e., non-simple) $s't'$-paths; consequently, also $d$ is redundant.
Moreover, $g$ and $h$, as well as the edges $(s',g)$ and $(f,h)$, are redundant, since they do not belong to any $s't'$-path (hence, the graph is not $s't'$-connected).
All other nodes and edges are irredundant.
Finally, the net is vulnerable, since there is an $st$-embedding of $\W$ into it: this is shown through the following node-disjoint subgraph homeomorphism from the graph in Fig.~\ref{fig:Wheat}(3) into the net in Fig.~\ref{fig:examplegraph}:
\begin{itemize}
\item $\phi$ maps $s$ to $a$, $u$ to $f$, $v$ to $b$, and $t$ to $t'$;
\item $\psi$ maps $(s,u)$ to $a \rightarrow e \rightarrow f$, $(s,v)$ to $a \rightarrow b$, $(u,v)$ to $f \rightarrow b$, $(v,t)$ to $f \rightarrow t'$, and $(u,t)$ to $b \rightarrow c \rightarrow t'$.
\end{itemize}
The path from $s'$ to $\phi(s) (= a)$ is not empty (it is made up by the single edge $(s',a)$), whereas the path from $t'$ to $\phi(t) (= t')$ is empty.
Notice that all paths involved are node-disjoint, excluding their extremes.

In \cite{CGS19} the authors devise an $O(nm^2)$ algorithm to determine if a given net is vulnerable by looking for a $\W$-embedding; this has been refined in \cite{MS23}, where  a faster $O(m^2)$ algorithm is provided. The latter one, that we call \MS\ (after the authors' names), is based on two alternated decompositions, the series one and the parallel one. They split the net sequentially until either each component is just a single edge, or they detect a homeomorphic copy of $\W$.
%
To obtain a linear time amortized cost for every addition, we start by refining the complexity analysis of \MS.

\begin{proposition}
\label{prop:MS}
Let $(G,s,t)$ be a net with $|V_G|=n$ and $|E_G|=m$. If $G$ has $k$ irredundant edges, then \MS$(G,s,t)$ costs $O(km)$.
\end{proposition}
Note that this result keeps the worst case to $O(m^2)$, if $k \in O(m)$.

To conclude, we remark that both the algorithms in \cite{CGS19} and \cite{MS23} have the interesting property that, if the net is not vulnerable, they also return the unique (as proved in \cite{CGS19}) \textit{Maximum Irredundant Subnet} (MIS) of that net, which is a series-parallel graph \cite{ChenEtal15,CinesiFull}.
Coming back to the graph of Fig.~\ref{fig:examplegraph}, its MIS is the subgraph obtained by removing $d$, $g$, $h$ and their incident edges;
clearly, it is not series-parallel, being the graph vulnerable. In contrast, the same graph but without the edge relating $(f,b)$ is non-vulnerable and, indeed, its MIS (still obtained by removing the same nodes and edges) is series-parallel.

Having the MIS of the current graph is the fundamental ingredient in what follows for checking the effects of adding a new edge to the net, as discussed in the next Sections.


\section{The Effects of Adding an Edge}

In this section, we study what can happen to a non-vulnerable graph $G= (V,E)$ after adding a new edge $e$ (which relates two existing vertices);
to use a simple notation, this will be written as $G \cup \{e\}$, to mean the pair $(V, E\cup\{e\})$.
We remark that in our setting it makes sense to consider only non-vulnerable graphs for the addition because, as a consequence of the characterisation given in \cite{CGS19}, if the original graph is already vulnerable, it will remain so even after any number of edge additions. So, in what follows, whenever we speak about the addition of an edge to a graph, we assume that the starting graph is non-vulnerable.

Being $G$ not vulnerable, we can assume to have its MIS (that we denote as $MIS(G)$) and this is series-parallel. If we add a new edge $e$ to $G$, we have three different situations:
\begin{itemize}
    \item $e$ is redundant in $G \cup \{e\}$: the net remains not vulnerable and the edge is not added to the MIS;
    \item $e$ is irredundant in $G \cup \{e\}$ and follows the series-parallel structure of $MIS(G)$: the net remains not vulnerable and $e$ is added to the MIS;
    \item $e$ is irredundant in $G \cup \{e\}$ but breaks the series-parallel structure of $MIS(G)$: the net becomes vulnerable.
\end{itemize}

We will see that these cases can be identified in linear time. The main problem is that $e$ can be linked with the MIS  not directly, but through other edges that were redundant in $G$ which, because of the new paths that can be created, can either become irredundant or stay redundant in $G \cup \{e\}$. Hence, when we add $e$, we have to look at all those redundant edges (which form what we call the {\em Redundant Connected Component} of $e$ -- see Definition~\ref{def:RCC} later on) and check its status in terms of vulnerability.

To do that, we first introduce a partial order relation on the nodes of the graph.

\newpage

\begin{definition}
\label{def:ordering}
    Let $G$ be an acyclic $st$-graph; given two nodes $a,b\in V_G$, 
    \begin{itemize}
    \item we write $a<b$, and say that {\em $a$ precedes $b$}, 
    if $a$ and $b$ are distinct and there is a simple $st$-path in $G$ where $a$ appears before $b$; 
    \item we say that they are {\em comparable} if $a=b$, or $a<b$, or $b<a$.
    \end{itemize}
    Notationally, we write: $A < b$ whenever every $a \in A$ precedes $b$; $a < B$ whenever $a$ precedes every $b \in B$; and
    $A < B$ whenever every $a \in A$ precedes every $b \in B$.
\end{definition}

We remark that ‘$<$' is defined only when the vertices belong to an acyclic graphs. Hence, in what follows, whenever we use it in the context of a generic (i.e., not necessarily acyclic) graph $G$, we actually mean that the relation holds for $MIS(G)$.

We then introduce the concept of \textit{linkable} nodes. The idea is that, when we add a new edge, a whole connected component can become irredundant; a first preliminary check to see if this may break the series-parallel structure of the original MIS is to consider the whole component only as a single edge.

\begin{definition}
\label{def:linkable}
    Given two nodes $a,b\in MIS(G)$, we say that $a$ and $b$ are {\em linkable} if the addition of an edge relating $(a,b)$ does not break the series-parallel structure of $MIS(G)$.
\end{definition}

\begin{wrapfigure}{r}{0.5\textwidth}
\vspace*{-.8cm}
\begin{tikzpicture}[every node/.style={circle,draw, minimum size=4mm, inner sep=2pt}]
\node (s) at (0,0) {$s$};
\node (a) at (1,1) {$a$};
\node (b) at (3,1) {$b$};
\node (c) at (2,0) {$c$};
\node (d) at (4,0) {$d$};
\node (t) at (6,0) {$t$};
\draw [->, bend left] (s) to (a);
\draw [->] (a) -- (b);
\draw [->, bend left] (b) to (d);
\draw [->] (s) -- (c);
\draw [->] (c) -- (d);
\draw [->] (d) -- (t);
\draw [->,bend left= 30, blue, dashed] (a) to (b);
\draw [->, bend left, blue, dashed] (s) to (d);
\draw [->, bend right = 40, red, dashed] (c) to (t);
\draw[->, bend left= 50, red, dashed] (a) to (t);
\end{tikzpicture}
\vspace*{-.8cm}
\end{wrapfigure}
For example, consider the MIS formed by the black edges in the graph to the right.
In this case, the pairs $a$ and $b$ and $s$ and $d$ are linkable (try and add the blue edges to the black ones), whereas the pairs $c$ and $t$ and $a$ and $t$ are not (try and add the red edges to the black ones).

\medskip

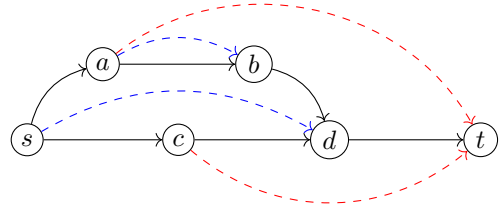
\begin{remark}\em
\label{rem:linkable}
Clearly, computing  whether two vertices are linkable can be done in $O(m)$, by using the linear-time algorithm for checking series-parallel (see \cite{VTL82}).
Moreover, whenever $a$ and $b$ are linkable, we can add any $s_{new} t_{new}$-series-parallel graph to $MIS(G)$ by identifying $a$ with $s_{new}$ and $b$ with $t_{new}$.
In the previous example, being $a$ and $b$ linkable, we can add the $ab$-series-parallel graph below to the left to the original graph and obtain the new (still series-parallel) graph below to the right:
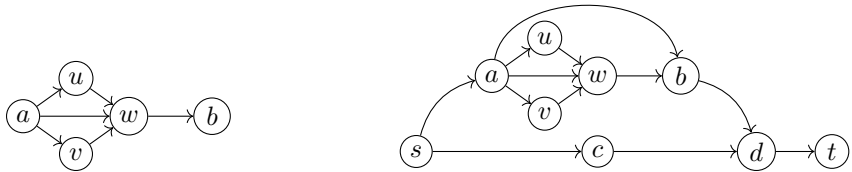
\vspace*{-.4cm}
\begin{figure}[H]
    \centering
\begin{tikzpicture}[every node/.style={circle,draw, minimum size=4mm, inner sep=2pt}]
\node (a) at (0,0) {$a$};
\node (u) at (.7,.5) {$u$};
\node (v) at (.7,-.5) {$v$};
\node (w) at (1.4,0) {$w$};
\node (b) at (2.5,0) {$b$};
\draw [->] (a) -- (u);
\draw [->] (a) -- (v);
\draw [->] (a) -- (w);
\draw [->] (u) -- (w);
\draw [->] (v) -- (w);
\draw [->] (w) -- (b);
\end{tikzpicture}
\hspace{2cm}
\begin{tikzpicture}[every node/.style={circle,draw, minimum size=4mm, inner sep=2pt}]
\node (s) at (0,0) {$s$};
\node (a) at (1,1) {$a$};
\node (u) at (1.7,1.5) {$u$};
\node (v) at (1.7,.5) {$v$};
\node (w) at (2.4,1) {$w$};
\node (b) at (3.5,1) {$b$};
\node (c) at (2.4,0) {$c$};
\node (d) at (4.5,0) {$d$};
\node (t) at (5.5,0) {$t$};
\draw [->, bend left] (s) to (a);
\draw [->, bend left=80] (a) to (b);
\draw [->] (a) -- (u);
\draw [->] (a) -- (v);
\draw [->] (a) -- (w);
\draw [->] (u) -- (w);
\draw [->] (v) -- (w);
\draw [->] (w) -- (b);
\draw [->, bend left] (b) to (d);
\draw [->] (s) -- (c);
\draw [->] (c) -- (d);
\draw [->] (d) -- (t);
\end{tikzpicture}
\end{figure}
\end{remark}

\medskip
We then formalize which redundant edges we have to take in consideration to potentially become irredundant after adding a new edge.
\begin{definition}
\label{def:RCC}
    The {\em Redundant Connected Component} (RCC) of the edge $e$, denoted $RCC(e)$, is the subgraph formed by 
    the nodes and edges 
    of any path containing 
    $e$ and such that: (1) it only uses redundant edges of $G$, and (2) 
    none of its internal nodes is in $MIS(G)$.
    We say that a node $\en$ is an {\em entry} of $RCC(e)$ if it belongs to $MIS(G)$ 
    and it is the starting extreme of any of the above paths.
%
    We say that a node $\ex$ is an {\em exit} of $RCC(e)$ if it belongs to $MIS(G)$ 
    and it is the ending extreme of any of the above paths.
\end{definition}

\begin{wrapfigure}{r}{0.4\textwidth}
\begin{tikzpicture}[every node/.style={circle,draw, minimum size=4mm, inner sep=2pt}]
\node (s) at (0,0) {$s$};
\node (a) at (1,1) {$a$};
\node (b) at (3,1) {$b$};
\node (c) at (2,0) {$c$};
\node (t) at (4,0) {$t$};
\node (d) at (2,1.7) {$d$};
\node (f) at (4,1) {$f$};
\node (g) at (0,1) {$g$};
\node (h) at (1,-1) {$h$};
\node (i) at (3,-1) {$i$};
\node (l) at (0,-1) {$l$};
\node (m) at (4,-1) {$m$};

\draw [->] (s) -- (a);
\draw [->] (a) -- (b);
\draw [->] (b) -- (t);
\draw [->] (s) -- (c);
\draw [->] (c) -- (t);
\draw [->, bend right, green] (t) to (f);
\draw [->, bend right, green] (f) to (b);
\draw [->, bend right, green] (b) to (d);
\draw [->, bend left, green] (g) to (a);
\draw [->, bend left, green] (m) to (i);
\draw [->, red] (s) to (h);
\draw [->, red] (i) to (t);
\draw [->, red] (c) to (h);
\draw [->, red] (l) to (h);
\draw [->, red,] (h) -- node[above,draw=none, text=black] {$e$} (i);
\end{tikzpicture}
\end{wrapfigure}
Let us consider the side graph, where the edge $e$ relating $(h,i)$ is added to $MIS(G)$ (depicted in black).
Here, the red edges are those of $RCC(e)$, the green ones are the redundant edges that are not in $RCC(e)$, the entries of $RCC(e)$ are $s$ and $c$, and its only exit is $t$.

The following definition summarizes the possible effects that the addition of an edge to a graph can have; these lead to the three main branches of our algorithm in  Section~\ref{sec:mainAlgo}. In defining the three cases, we have to take into account the possible effects of adding the whole $RCC(e)$ while adding $e$ to the MIS.

\begin{definition}
\label{def:I-I-D}
    Let $e$ be the edge to be added, with $N$ and $X$ be the entries and the exits of $RCC(e)$.
    We say that the addition of $e$ is 
    \begin{itemize}
    \item {\em impactless}, if $X \leq N$ in $G$;
    \item {\em impactful}, if there exists exactly one $(\en,\ex)\in N \times X$ such that in $G$ it holds that $$\en<\ex\ \land\ \en\text{ and } \ex\text{ are linkable}\ \land\ \forall (\en',\ex') \in N \times X\,.\, ((\en' \neq \en \lor \ex' \neq \ex) \Rightarrow \en'\ge \ex')$$ 
    In this case, we say that the pair $(\en,\ex)$ is the {\em entry-exit pair} for $e$;
    \item {\em destructive}, otherwise.
    \end{itemize}
\end{definition}

The names we chose for the different possible cases of edge addition are justified by the following three Lemmas; these are the main results for building our algorithm and will be the main ingredients for the proof of the correctness theorem (Theorem~\ref{thm:correct}).

The first Lemma states that the first case of Definition~\ref{def:I-I-D} implies that the added edge is redundant in the resulting net;
so, its addition does not change the assumed non-vulnerability of the original graph and no check must be carried out on its RCC
(hence, adding $e$ has no impact on the amortized cost of the incremental algorithm).

\begin{lemma}\label{redred}
If the addition of $e$ is impactless, then $e$ is redundant in $G\cup \{e\}$ and ${MIS(G\cup\{e\})}=MIS(G)$.
\end{lemma}

\begin{figure}
    \centering

\begin{tikzpicture}[every node/.style={circle,draw, minimum size=4mm, inner sep=2pt}]
\node (s) at (0,0) {$s$};
\node (a) at (1,0) {$a$};
\node (b) at (2,0) {$b$};
\node (t) at (3,0) {$t$};
\node (c) at (1,-1) {$c$};
\node (d) at (2,-1) {$d$};

\draw [->] (s) -- (a);
\draw [->] (a) -- (b);
\draw [->] (b) -- (t);
\draw [->, red] (c) to (s);
\draw [->, red] (c) to (a);
\draw [->, red] (b) to (d);
\draw [->, red] (t) to (d);
\draw [->, red,] (d) -- node[above,draw=none, text=black] {$e$} (c);
\end{tikzpicture}
    \caption{An example where the addition of $e$ is impactless}
    \label{fig:impacteless}
\end{figure}
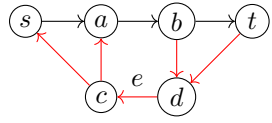

As a concrete case, consider the graph in Fig.~\ref{fig:impacteless}, where the edge $e$ relating $(d,c)$ is added; all other edges already belong to $G$ (we depict in  black the edges in $MIS(G)$ and in red those of $RCC(e)$).
The entries are $b$ and $t$, and the exits are $s$ and $a$; the addition of $e$ is impactless, since $\{s,a\} < \{b,t\}$ in $G$. Notice that $e$ (together with $RCC(e)$) is redundant in the new graph.
\medskip

The next two Lemmas apply when the added edge respects the series-parallel structure of the MIS. 
If we are in the second case of Definition~\ref{def:I-I-D}, this implies that we can add all the irredundant part of $RCC(e)$ to $MIS(G)$ without breaking its series-parallel structure and, actually, this gives the MIS of the new net.
However, $RCC(e)$ may contain a homeomorphic copy of $\W$ and so $\MS$ must be run on it (we will show later in Lemma~\ref{irrvuln} that this suffices); hence, in this case the addition of $e$ has an impact on the amortized analysis.

\begin{lemma}\label{irredirred}
    If the addition of $e$ is impactful, then $e$ is irredundant in $G\cup \{e\}$. Moreover, if $(\en,\ex)$ is its entry-exit pair, all the edges of $RCC(e)$ that become irredundant in $G \cup \{e\}$ (seen as an $\en\ex$-net) are a new parallel component that can be added to $MIS(G)$ to form $MIS(G \cup \{e\})$.
\end{lemma}

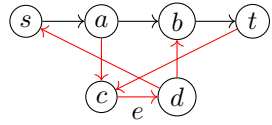
\begin{figure}
    \centering
\begin{tikzpicture}[every node/.style={circle,draw, minimum size=4mm, inner sep=2pt}]
\node (s) at (0,0) {$s$};
\node (a) at (1,0) {$a$};
\node (b) at (2,0) {$b$};
\node (t) at (3,0) {$t$};
\node (c) at (1,-1) {$c$};
\node (d) at (2,-1) {$d$};

\draw [->] (s) -- (a);
\draw [->] (a) -- (b);
\draw [->] (b) -- (t);
\draw [->, red] (d) to (s);
\draw [->, red] (a) to (c);
\draw [->, red] (d) to (b);
\draw [->, red] (t) to (c);
\draw [->, red,] (c) -- node[below,draw=none, text=black] {$e$} (d);
\end{tikzpicture}
    \caption{An example where the addition of $e$ is impactful}
    \label{fig:impactful}
\end{figure}

To see an example of impactful addition, consider the graph in Fig.~\ref{fig:impactful} (with the same graphical conventions as those in Fig.~\ref{fig:impacteless}).
The entries now are $a$ and $t$ and the exits are $s$ and $b$. Consider the pair $a$ and $b$; it is the entry-exit pair for $e$. Indeed, the two vertices are linkable and $a<b$ in $G$; moreover, $a>s$, $t>s$ and $t>b$. Thus, the addition of $e$ is impactful and, indeed, $e$ is irredundant in the new graph; moreover, $MIS(G \cup \{e\})$ has the new parallel component $a \rightarrow c \rightarrow d \rightarrow b$.


Finally, the third Lemma states that, if we are in the third case of Definition~\ref{def:I-I-D}, then the added edges (i.e., $e$ and the irredundant part of its RCC) break the series-parallel structure of the original MIS;
thus, the addition of $e$ destroys the original non-vulnerability of $G$.

\begin{lemma}\label{vulnvuln}
    If the addition of $e$ is destructive, then $G\cup \{e\}$ is vulnerable.
\end{lemma}

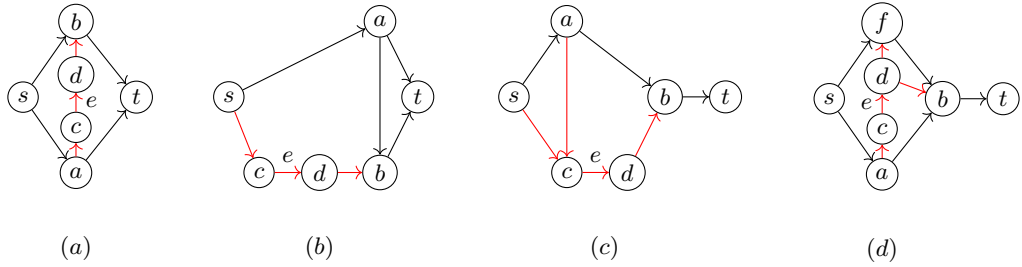
\begin{figure}[t]
\begin{tabular}{cccc}
\begin{minipage}{0.16\textwidth}
\center
\begin{tikzpicture}[every node/.style={circle,draw, minimum size=4mm, inner sep=2pt}]
\node (s) at (0,0) {$s$};
\node (a) at (.7,-1) {$a$};
\node (b) at (.7,1) {$b$};
\node (t) at (1.5,0) {$t$};
\node (c) at (.7,-0.4) {$c$};
\node (d) at (.7,0.3) {$d$};
\node[draw=none] (1) at (.7,-2) {$(a)$};

\draw [->] (s) -- (a);
\draw [->] (s) to (b);
\draw [->] (a) -- (t);
\draw [->] (b) to (t);
\draw [->, red] (a) to (c);
\draw [->, red] (d) to (b);
\draw [->, red,] (c) -- node[right,draw=none, text=black] {$e$} (d);
\end{tikzpicture}
\end{minipage}
&
\begin{minipage}{0.24\textwidth}
\center
\begin{tikzpicture}[every node/.style={circle,draw, minimum size=4mm, inner sep=2pt}]
\node (s) at (0,0) {$s$};
\node (a) at (2,1) {$a$};
\node (b) at (2,-1) {$b$};
\node (t) at (2.5,0) {$t$};
\node (c) at (0.4,-1) {$c$};
\node (d) at (1.2,-1) {$d$};
\node[draw=none] (2) at (1.2,-2) {$(b)$};

\draw [->] (s) -- (a);
\draw [->] (a) to (t);
\draw [->] (a) -- (b);
\draw [->] (b) -- (t);
\draw [->, red] (s) to (c);
\draw [->, red] (d) to (b);
\draw [->, red,] (c) -- node[above,draw=none, text=black] {$e$} (d);
\end{tikzpicture}
\end{minipage}
&
\begin{minipage}{0.26\textwidth}
\center
\begin{tikzpicture}[every node/.style={circle,draw, minimum size=4mm, inner sep=2pt}]
\node (s) at (0,0) {$s$};
\node (a) at (.7,1) {$a$};
\node (b) at (2,0) {$b$};
\node (t) at (2.8,0) {$t$};
\node (c) at (.7,-1) {$c$};
\node (d) at (1.5,-1) {$d$};
\node[draw=none] (3) at (1.2,-2) {$(c)$};

\draw [->] (s) -- (a);
\draw [->] (a) -- (b);
\draw [->] (b) -- (t);
\draw [->, red] (s) to (c);
\draw [->, red] (a) to (c);
\draw [->, red] (d) to (b);
\draw [->, red,] (c) -- node[above,draw=none, text=black] {$e$} (d);
\end{tikzpicture}
\end{minipage}
&
\begin{minipage}{0.24\textwidth}
\center
\begin{tikzpicture}[every node/.style={circle,draw, minimum size=4mm, inner sep=2pt}]
\node (s) at (0,0) {$s$};
\node (a) at (.7,-1) {$a$};
\node (f) at (.7,1) {$f$};
\node (c) at (.7,-.4) {$c$};
\node (d) at (.7,.3) {$d$};
\node (b) at (1.5,0) {$b$};
\node (t) at (2.3,0) {$t$};
\node[draw=none] (4) at (.7,-2) {$(d)$};

\draw [->] (s) to (a);
\draw [->] (a) to (b);
\draw [->] (b) -- (t);
\draw [->] (s) to (f);
\draw [->] (f) to (b);
\draw [->, red] (d) to (f);
\draw [->, red] (a) to (c);
\draw [->, red] (d) to (b);
\draw [->, red,] (c) -- node[left,draw=none, text=black] {$e$} (d);
\end{tikzpicture}
\end{minipage}
\end{tabular}
\caption{Four graphs exemplifying the cases for Lemma~\ref{vulnvuln}.}
\label{fig:red}
\vspace{-3mm}
\end{figure}

\iffull
\else
The proof of the previous Lemma shows that
\fi
a destructive addition can occur for four different possible reasons; in Figure~\ref{fig:red}, we provide an exemplification of all these cases. In all of them, $e$ and its redundant connected component contribute to the creation of a homeomorphic copy of $\W$, making the new graph vulnerable.
The first case is exemplified by Fig.~\ref{fig:red}(a): in this case, the addition is destructive because the entry $a$ and the exit $b$ are not comparable in the original graph; so, the addition can be neither impactless nor impactful.
The second case is exemplified by Fig.~\ref{fig:red}(b): here, the entry $s$ and the exit $b$ are such that $s<b$ (hence, the addition cannot be impactless), but $s$ and $b$ are not linkable (hence, the addition cannot be impactful).
For the third case, consider Fig.~\ref{fig:red}(c): here we have two entries, $s$ and $a$, that both precede the exit $b$, i.e. $s<b$ and $a<b$, and clearly these two pairs are both linkable.
Finally, consider Fig.~\ref{fig:red}(d): here, the entry $a$ and the exit $b$ form the only linkable pair such that $a<b$; however, $a$ and the other exit $f$ are not comparable in the original graph.

The last ingredient for the correctness of our algorithm is Lemma~\ref{irrvuln}, whose proof needs a preliminary Proposition,
which allows us to check, in case of a non-destructive addition, only the edges in $RCC(e)$ and to ignore all the other redundant edges of the graph.

\begin{proposition}\label{Prop:RE}
If the addition of $e$ to $G$ is non-destructive, then every redundant edge in $G$ which does not belong to $RCC(e)$ remains redundant in $G \cup \{e\}$.
\end{proposition}

\begin{wrapfigure}{r}{0.35\textwidth}
\vspace*{-.25cm}
\begin{tikzpicture}[every node/.style={circle,draw, minimum size=4mm, inner sep=2pt}]
\node (s) at (0,0) {$s$};
\node (a) at (1,0) {$a$};
\node (b) at (2,0) {$b$};
\node (c) at (1,1) {$c$};
\node (d) at (2,1) {$d$};
\node (e) at (3,0) {$e$};
\node (t) at (4,0) {$t$};
\draw [->] (s) -- (a);
\draw [->] (a) -- (b);
\draw [->] (b) -- (e);
\draw [->] (e) -- (t);
\draw[->, red] (a) to (c);
\draw[->, red] (d) to (e);
\draw [->, red] (c) to node[below,draw=none, text=black] {$e$} (d);

\draw [->, blue, bend left] (e) to (b);
\draw [->, blue, bend left=40] (e) to (s);
\draw [->, blue, bend right] (e) to (d);
\end{tikzpicture}
\end{wrapfigure}
To visualize the situation of the previous Proposition, consider the side graph.
Here, we are adding the edge $e$ relating $(c,d)$, the black edges form $MIS(G)$ whereas the red ones are $RCC(e)$. The blue edges are those not in $RCC(e)$ which are redundant in $G$: all of them remain redundant after adding $e$.

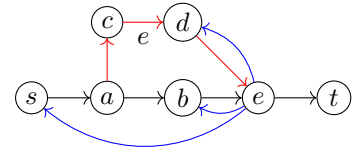
\begin{wrapfigure}{r}{0.35\textwidth}
\vspace*{-.25cm}
\begin{tikzpicture}[every node/.style={circle,draw, minimum size=4mm, inner sep=2pt}]
\node (s) at (0,0) {$s$};
\node (a) at (1,0) {$a$};
\node (b) at (2,1) {$b$};
\node (t) at (3,0) {$t$};
\node (c) at (2,-1) {$c$};

\draw [->] (s) -- (a);
\draw [->] (a) -- (b);
\draw [->] (a) -- (c);
\draw [->] (b) -- (t);
\draw [->] (c) -- (t);
\draw [->, red,bend left] (s) to node[above,draw=none, text=black] {$e$} (b);
\draw [->, bend left, blue] (b) to node[below,draw=none, text=black] {$e'$} (a);
\end{tikzpicture}
\vspace*{-.4cm}
\end{wrapfigure}
In contrast, in the side  graph, when we add $e$, such an edge is irredundant in the resulting net and makes $e'$ irredundant too, even though $e' \not\in RCC(e)$; however, this example does not contradict the previous Proposition, since the addition of $e$ is destructive (in fact, the new net is vulnerable).

\begin{corollary}
\label{cor:newMIS}
    If the addition of $e$ to $G$ is impactful, then the MIS of $G \cup \{e\}$ is $MIS(G)$ together with all the edges and vertices of $RCC(e)$ that become irredundant after the addition of $e$.
\end{corollary}

\begin{lemma}\label{irrvuln}
Let the addition of $e$ be impactful with entry-exit pair $(\en,\ex)$; then, $G\cup \{e\}$ is vulnerable if and only if $RCC(e)$  (seen as an $\en\ex$-net) is vulnerable.
\end{lemma}

\begin{wrapfigure}{r}{0.35\textwidth}
\vspace*{-.5cm}
\begin{tikzpicture}[every node/.style={circle,draw, minimum size=4mm, inner sep=2pt}]
\node (s) at (0,0) {$s$};
\node (a) at (1,0) {$a$};
\node (b) at (2,0) {$b$};
\node (c) at (3,0) {$c$};
\node (t) at (4,0) {$t$};
\node (d) at (1,-1) {$d$};
\node (e) at (2,-.5) {$h$};
\node (f) at (2,-1.5) {$f$};
\node (g) at (3,-1) {$g$};
\draw [->] (s) -- (a);
\draw [->] (a) -- (b);
\draw [->] (b) -- (c);
\draw [->] (c) -- (t);
\draw[->, red] (d) to (e);
\draw[->, red] (d) to (f);
\draw[->, red] (e) to (g);
\draw[->, red] (f) to (g);
\draw[->, red] (e) to (f);
\draw [->, red, bend right] (a) to node[left,draw=none, text=black] {$e$} (d);
\draw [->, red, bend right] (g) to (c);
\end{tikzpicture}
\vspace*{-.4cm}
\end{wrapfigure}
The scenario of the previous Lemma is exemplified by the side graph. When adding the edge $e$ relating $(a,d)$, the red edges are $RCC(e)$ and the entry-exit pair is $(a,c)$. Note that $a$ and $c$ are linkable, but the whole graph is vulnerable since the $ac$-net in red is vulnerable.


\section{The Incremental Algorithm}
The algorithm we propose exploits the information returned by running the \MS\ algorithm by \cite{MS23} on a non-vulnerable net (i.e. its MIS, which is series-parallel). Assuming that the starting net is not vulnerable, we analyze the RCC of the added edge and recognize the type of addition, among the three ones of the previous section. In the impactful case, we need to run \MS\ on the RCC, but we can achieve an amortized cost of the whole algorithm of $O(m)$.

The incremental vulnerability check is obtained by combining 5 algorithms. First, we rely on the algorithm \MS\ by \cite{MS23}. Then, we develop an algorithm for calculating the RCC of the new edge $e$ (Algorithm~\ref{alg:RCC}) and an algorithm for ascertaining whether the addition of $e$ is impactless/impactful/destructive (Algorithm~\ref{alg:R-I-V}); the latter one relies on another algorithm that adapts the series-parallel recognition algorithm by \cite{VTL82} to efficiently calculate the minimum vertex greater (w.r.t. the ordering defined in Definition~\ref{def:ordering}) than a given set of vertices (Algorithm~\ref{alg:minFoll}).
Finally, we can provide our incremental algorithm (Algorithm~\ref{alg:vulnGue}).

\subsection{Calculating the Redundant Connected Component}

We start with Algorithm~\ref{alg:RCC} which finds $RCC(e)$ and also returns all the entries $\{\en_i\}_i$ as $\EN$ and all the exits $\{\ex_j\}_j$ as $\EX$. We assume that all edges and nodes of $MIS(G)$ have already been identified in $G$. Since the algorithm consists of two visits, its cost is linear.

\begin{algorithm}[t]
    \caption{$RCC(G,e)$}
    \begin{algorithmic}[1]
        \Require All edges and nodes of $MIS(G)$ are marked in $G$
        \State $rcc=\{e\}$, $\EN = \EX = \emptyset$
        \State Do a backward visit from the starting vertex of $e$, stopping at $MIS(G)$: during this visit, add all the touched edges to $rcc$ and
        all the touched nodes of $MIS(G)$ to $\EN$.
        \State Do a forward visit from the ending vertex of $e$, stopping at $MIS(G)$: during this visit, add all the touched edges to $rcc$ and
        all the touched nodes of $MIS(G)$ to $\EX$.
        \State\Return $(rcc,\EN,\EX)$
    \end{algorithmic}
\label{alg:RCC}
\end{algorithm}

    \begin{proposition}
    \label{prop:costRCC}
$RCC(G,e)$ costs $O(m)$.
    \end{proposition}

\subsection{Recognizing the Type of Edge Addition}

Algorithm~\ref{alg:Salvo} checks the conditions in Definition~\ref{def:I-I-D} to determine whether the addition of a new edge $e$ is impactless, impactful, or destructive, by keeping the overall complexity in $O(m)$. Since there are some cases where we need to check whether a set of nodes $A$ precedes another set $B$, we need a preliminary proposition to efficiently do this.

\begin{proposition}
\label{prop:Tu}
Let $G$ be a series-parallel graph and $A, B\subseteq V_G$.
Then, $A\leq B$ if and only if there exists $w\in V_G$ such that $A\leq w \leq B$.
\end{proposition}

\begin{algorithm}[t]
    \caption{$\mathit{TypeOfAdd}(G, mis, N, X)$}
    \label{alg2}    \label{alg:Salvo}
    \begin{algorithmic}[1]
        \If{$N=\varnothing \lor X=\varnothing$} \label{algS:check-empty}
        \State\Return {\sc Impactless} \label{algS:empty}
        \EndIf
        \State $w = MinFoll(mis,X)$ \label{algS:Tleqw}  
        \State Let $N'$  be the largest subset of $N$ s.\,t. $w\leq N'$ \label{algS:impless1}
		\If{$|N'| = |N|$}  
		\State {\bf return} {\sc Impactless} \label{algS:impless2}
		\EndIf
		\If{$|N'|<|N| - 1$} \label{algS:destr1}
		\State {\bf return} {\sc Destructive} \label{algS:destr2}
		\EndIf
		\State let $\{\en\} = N\setminus N'$ \label{algS:impactful}
		\State Let $X'$ be the largest subset of $X$ s.\,t. $X' \leq \en$ \label{algS:X'}
		\If{$|X'| \neq |X| - 1$} \label{algS:destr3}
		  \State {\bf return} {\sc Destructive} \label{algS:destr4}
		\EndIf
		\State let $\{\ex\} = X\setminus X'$\label{algS:ex}
		\If{$\en < \ex \land linkable(\en,\ex) \land X'\leq N  \land X \leq N'$} \label{algS:check}
		  \State {\bf return} ({\sc Impactful}, \en, \ex)\label{algs:returnimpactful}
		  \Else 
		  \State {\bf return} {\sc Destructive}\label{algS:destr5}
		\EndIf
    \end{algorithmic}
    \label{alg:R-I-V}
    \end{algorithm}
    
We defer the development of a linear algorithm for finding the node $w$ of the previous Proposition to Algorithm~\ref{alg:minFoll};
so, let us first focus on Algorithm~\ref{alg:Salvo}.
This latter algorithm is invoked on a graph $G$ (that is non-vulnerable) for which we already know its MIS (contained in the variable $mis$) and the entries/exits of $RCC(e)$ (viz., $N$ and $X$), where $e$ is the edge that is added.\footnote{Notice that, for this algorithm to properly work, we do not need to know which is the added edge; just its entries and exits suffice.}
If either $N$ or $X$ is empty, then the edge we are adding does not lie on any $st$-path; so, the addition is impactless (lines \ref{algS:check-empty}--\ref{algS:empty}).
If this is not the case, we check whether $X \leq N$ in $MIS(G)$; by Proposition \ref{prop:Tu}, 
this is equivalent to finding any node $w$ such that $X\leq w\leq N$. 

To this aim, we first invoke Algorithm~\ref{alg:minFoll} and find the minimal (w.r.t. ‘$\leq$') $w$ that follows $X$ in $MIS(G)$, i.e. such that $X \leq w$ (line \ref{algS:Tleqw}); this always exists (at most, it is $t$).
Then, we determine the largest set $N' \subseteq N$ of nodes reachable from $w$. If $N'=N$, the algorithm says that the addition of $e$ is impactless, because $X \leq N$. 

If exactly one $\en \in N$ is not reachable from $w$, we explore whether $\en$ can be the entry of the entry-exit pair for $e$ (and, so, the addition of $e$ is impactful). 
To this end, we determine the largest subset $X'$ of nodes from $X$ that precede $\en$.
Again, if exactly one node $\ex \in X$ does not precede $\en$, this can potentially be the exit of the entry-exit pair for $e$; 
thus, we must check that (1) $\en$ precedes $\ex$ and they are linkable, and (2) all other pairs of entry $\en'\in N$ and exit $\ex' \in X$ are such that $\ex' \leq \en'$ (this is done in line \ref{algS:check}). If all these checks succeed, Algorithm \ref{alg:Salvo} says that the addition of $e$ is impactful; in all other cases, Algorithm \ref{alg:Salvo} says that the addition of $e$ is destructive.

Since the checks in lines 4, 10 and 14 rely
on 
visits and the cost of $MinFoll$ is linear (see Lemma~\ref{minFollCorrect} later on), its overall cost is linear too. 

\begin{lemma}[Correctness and Complexity]
    \label{prop_cost-R-I-V}
Algorithm~\ref{alg:Salvo} correctly returns the type of addition of $e$ in $G$
and costs $O(m)$.
\end{lemma}

To see how Algorithm \ref{alg:Salvo} works, let us consider the examples in Figures \ref{fig:impacteless}, \ref{fig:impactful}, \ref{fig:red}:
\begin{itemize}
 \item For the network in Fig.~\ref{fig:impacteless}, we have $N=\{b,t\}$ and $X=\{s,a\}$ (these are obtained by running Algorithm \ref{alg:RCC} on such a graph and its MIS). 
Thus, Algorithm \ref{alg:Salvo} sets $w=a$ in line \ref{algS:Tleqw} and $N'=\{b,t\}$ in line \ref{algS:impless1}; since $|N|=|N'|$, in line \ref{algS:impless2} the algorithm  correctly returns {\sc Impactless}.

\item For the network in Fig.~\ref{fig:impactful}, we have $N=\{a,t\}$ and $X=\{s,b\}$; then, Algorithm \ref{alg:Salvo} sets $w=b$ in line \ref{algS:Tleqw}, $N'=\{t\}$ in line \ref{algS:impless1}, $\en=a$ in line \ref{algS:impactful}, $X'=\{s\}$ in line \ref{algS:X'}, and $\ex=b$ in line \ref{algS:ex}. Since $a<b$, $a$ and $b$ are linkable, $\{s\}\le\{a,t\}$ and $\{s,b\}\le\{t\}$, in line \ref{algs:returnimpactful} the algorithm correctly returns {\sc Impactful}.

\item For the network in Fig.~\ref{fig:red}(c), we have $N=\{s,a\}$ and $X=\{b\}$; then, Algorithm \ref{alg:Salvo} sets $w=b$ in line \ref{algS:Tleqw} and $N'=\emptyset$ in line \ref{algS:impless1}. In line \ref{algS:destr2} the algorithm correctly returns {\sc Destructive}, since $|N'|=0<1=|N|-1$.

\item For the network in Fig.~\ref{fig:red}(d), we have $N=\{a\}$ and $X=\{b,f\}$; then, Algorithm \ref{alg:Salvo} sets $w=b$ in line \ref{algS:Tleqw},
$N'=\emptyset$ in line \ref{algS:impless1}, $\en=a$ in line \ref{algS:impactful}, and $X'=\emptyset$ in line \ref{algS:X'}. Since $|X'|=0<1=|X|-1$, in line \ref{algS:destr4} the algorithm correctly returns {\sc Destructive}.

\item For the network in Fig.~\ref{fig:red}(a), we have $N=\{a\}$ and $X=\{b\}$; then, Algorithm \ref{alg:Salvo} sets $w=b$ in line \ref{algS:Tleqw},  $N'=\emptyset$ in line \ref{algS:impless1}, $\en=a$ in line \ref{algS:impactful}, $X'=\emptyset$ in line \ref{algS:X'} and $\ex=b$ in line \ref{algS:ex}. Since $a$ and $b$ are not linkable, in line \ref{algS:destr5} the algorithm correctly returns {\sc Destructive}.

\item For the network in Fig.~\ref{fig:red}(b), the behavior is analogous to the previous case.
\end{itemize}

We are left with the task of linearly finding a node $w$ that follows all the nodes of a given set $A$.
This is what Algorithm \ref{alg:minFoll} does: it finds a minimal (w.r.t. ‘$\leq$') node $w$ such that $A\leq w$,
by adapting the algorithm from \cite{VTL82} for deciding whether a given graph is two-terminal series-parallel or not.
Indeed, lines~\ref{startTarjan}--\ref{endTarjan} (excluding lines 10 and 11) are the linear-time algorithm by Valdes et al. \cite{VTL82} to recognize whether a given graph is series-parallel; since we assume that we are running $MinFoll$ over such a kind of graph, we are also sure that we always find a $v$ with in-degree and out-degree equal to 1, unless the graph is one single edge from $s$ to $t$.

\begin{algorithm}[t]
    \caption{$MinFoll(G, A)$ \\ Returns the minimum (w.r.t.\,‘$\leq$') vertex that follows $A \subseteq V$ in the $st$-series-parallel graph $G = (V,E)$} 
    \label{alg:minFoll}
    \begin{algorithmic}[1]
        \ForAll{$v\in V$}
    	\If{$v \in A$}
	\State $np_A[v] = 1$
	\Else\ $np_A[v] = 0$
	\EndIf
        \EndFor
	\State Let $\hat G = (\hat V, \hat E)$, with $\hat V = V$ and $\hat E = E$ \label{startTarjan}
        \State Remove from $\hat E$ all parallel edges
        \While{$\hat G \neq s \rightarrow t$} 
        \State Let $v \in \hat V$ be such that $in\_degree_{\hat G}(v) = out\_degree_{\hat G}(v) = 1$
        \State Let $(u , v)$ and $(v,w)$ belong to $\hat E$
        \State $np_A[w]$ += $np_A[v]$
        \State $\mathit{ref}[v] = u$
        \State Cancel $v$ from $\hat V$ and $(u , v)$ and $(v,w)$ from $\hat E$
        \If{$(u,w) \not\in \hat E$}
        \State add $(u,w)$ to $\hat E$ \label{endTarjan}
        \EndIf
        \EndWhile
        \State $prec_A[s] = np_A[s]$
        \State $prec_A[t] = |A|$
        \ForAll{$v \in V$ touched during a topological visit of $G$}
        \If{$v \not\in \{s,t\}$}
        \State $prec_A[v] = np_A[v] + prec_A[\mathit{ref}[v]]$
        \EndIf
        \If{$prec_A[v] = |A|$}
        \State \Return $v$
        \EndIf
        \EndFor
    \end{algorithmic}
    \end{algorithm}

We exploit three data structures:
\begin{itemize}
\item $\mathit{ref}[v]$ is the immediate predecessor of $v$ in $\hat G$ at the moment of its cancellation (see line 12);
\item $np_A[v]$ counts the elements of $A$ that are ‘covered' by $v$ in all paths of the original graph $G$ that start from $\mathit{ref}[v]$ and arrive into $v$ (excluding $\mathit{ref}[v]$, if this vertex belongs to $A$);
\item $prec_A[v]$ counts the elements of $A$ that are ‘covered' by $v$; this number is obtained in the final topological visit of the original (series-parallel) graph $G$ by adding $np_A[v]$ and $prec_A[\mathit{ref}[v]]$.
\end{itemize}

\begin{lemma}
\label{lem:npA}
If $\mathit{ref}[v] = u$, then $np_A[v] = |\{a \in A \,:\, u < a \leq v \text{ in } G\}|$.
\end{lemma}

\begin{proposition}
\label{prop:prec}
If $prec_A[v]$ is defined, then $prec_A[v] = |\{a \in A \,:\, a \leq v \text{ in } G\}|$.
\end{proposition}

The overall cost of Algorithm~\ref{alg:minFoll} is linear, since the {\bf while} loop is essentially the linear-time algorithm by \cite{VTL82}, together with two extra constant-time operations, followed by a topological visit, that is linear too.
By this and Propositions \ref{prop:Tu} and \ref{prop:prec}, we have the following.

\begin{lemma}[Correctness and Complexity]
\label{minFollCorrect}
Let $G$ be $st$-series-parallel and $A \subseteq V_G$; if $MinFoll(G, A)$ returns $w$, then $A \leq w$ and, for all $w'$ such that $A \leq w'$, it holds that $w \leq w'$. Moreover, its cost is $O(m)$.
\end{lemma}

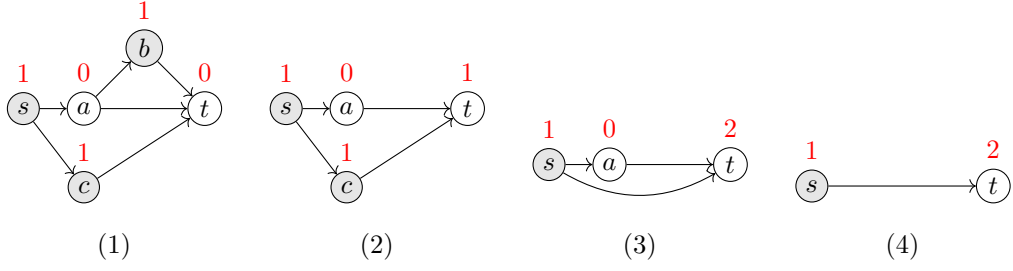
\begin{figure}[t]
\centering

\begin{minipage}[t]{0.24\textwidth}
\centering
\begin{tikzpicture}[scale=0.8, every node/.style={circle,draw, minimum size=4mm, inner sep=2pt}]
\node[fill=gray!20, label=above:\textcolor{red}{1}] (s) at (0,0) {$s$};
\node[label=above:\textcolor{red}{0}] (a) at (1,0) {$a$};
\node[fill=gray!20,label=above:\textcolor{red}{1}] (b) at (2,1) {$b$};
\node[label=above:\textcolor{red}{0}] (t) at (3,0) {$t$};
\node[fill=gray!20,label=above:\textcolor{red}{1}] (c) at (1,-1.3) {$c$};

\draw[->] (s)--(a);
\draw[->] (a)--(b);
\draw[->] (a)--(t);
\draw[->] (b)--(t);
\draw[->] (s)--(c);
\draw[->] (c)--(t);
\end{tikzpicture}

\vspace{2mm}
\parbox[c][1.5em][c]{\linewidth}{\centering (1)}
\end{minipage}
\begin{minipage}[t]{0.24\textwidth}
\centering
\begin{tikzpicture}[scale=0.8, every node/.style={circle,draw, minimum size=4mm, inner sep=2pt}]
\node[fill=gray!20,label=above:\textcolor{red}{1}] (s) at (0,0) {$s$};
\node[label=above:\textcolor{red}{0}] (a) at (1,0) {$a$};
\node[label=above:\textcolor{red}{1}] (t) at (3,0) {$t$};
\node[fill=gray!20,label=above:\textcolor{red}{1}] (c) at (1,-1.3) {$c$};

\draw[->] (s)--(a);
\draw[->] (a)--(t);
\draw[->] (s)--(c);
\draw[->] (c)--(t);
\end{tikzpicture}

\vspace{2mm}
\parbox[c][1.5em][c]{\linewidth}{\centering (2)}
\end{minipage}
\begin{minipage}[t]{0.24\textwidth}
\centering
\begin{tikzpicture}[scale=0.8, every node/.style={circle,draw, minimum size=4mm, inner sep=2pt}]
\node[fill=gray!20,label=above:\textcolor{red}{1}] (s) at (0,0) {$s$};
\node[label=above:\textcolor{red}{0}] (a) at (1,0) {$a$};
\node[label=above:\textcolor{red}{2}] (t) at (3,0) {$t$};

\draw[->] (s)--(a);
\draw[->] (a)--(t);
\draw[->] (s) to[bend right] (t);
\end{tikzpicture}

\vspace{2mm}
\parbox[c][1.5em][c]{\linewidth}{\centering (3)}
\end{minipage}
\begin{minipage}[t]{0.24\textwidth}
\centering
\begin{tikzpicture}[scale=0.8, every node/.style={circle,draw, minimum size=4mm, inner sep=2pt}]
\node[fill=gray!20,label=above:\textcolor{red}{1}] (s) at (0,0) {$s$};
\node[label=above:\textcolor{red}{2}] (t) at (3,0) {$t$};

\draw[->] (s)--(t);
\end{tikzpicture}

\vspace{2mm}
\parbox[c][1.5em][c]{\linewidth}{\centering (4)}
\end{minipage}

\caption{Calculation of $np_A$ ({\bf while} loop of Algorithm~\ref{alg:minFoll}); nodes of $A$ are in grey and values of $np_A$ are in red.}
\label{ex:alg3}
\end{figure}

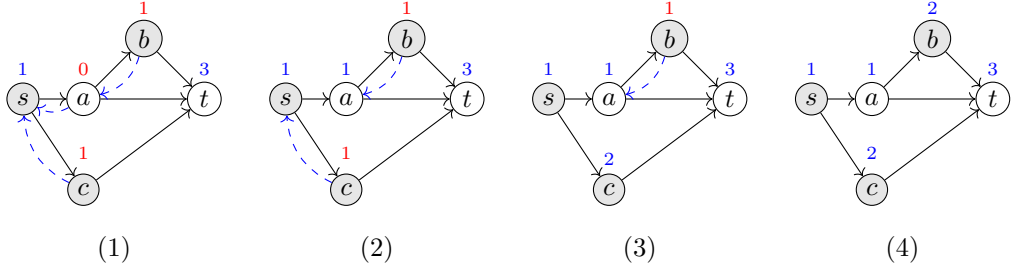
\begin{figure}[t]
\centering

\begin{minipage}[t]{0.24\textwidth}
\centering
\begin{tikzpicture}[
scale=0.8,
vertex/.style={circle,draw, minimum size=4mm, inner sep=2pt},
labelnode/.style={draw=none, fill=none, font=\scriptsize}
]

\node[fill=gray!20,vertex] (s) at (0,0) {$s$};
\node[labelnode,text=blue] at (0,0.5) {1};

\node[vertex] (a) at (1,0) {$a$};
\node[labelnode,text=red] at (1,0.5) {$0$};

\node[fill=gray!20,vertex] (b) at (2,1) {$b$};
\node[labelnode,text=red] at (2,1.5) {$1$};

\node[vertex] (t) at (3,0) {$t$};
\node[labelnode,text=blue] at (3,0.5) {3};

\node[fill=gray!20,vertex] (c) at (1,-1.5) {$c$};
\node[labelnode,text=red] at (1,-0.9) {$1$};

\draw[->] (s)--(a);
\draw[->] (a)--(b);
\draw[->] (a)--(t);
\draw[->] (b)--(t);
\draw[->] (s)--(c);
\draw[->] (c)--(t);
\draw [->,bend left= 30, blue, dashed] (a) to (s);
\draw [->,bend left= 30, blue, dashed] (c) to (s);
\draw [->,bend left= 30, blue, dashed] (b) to (a);

\end{tikzpicture}

\vspace{2mm}
\parbox[c][1.5em][c]{\linewidth}{\centering (1)}
\end{minipage}
\begin{minipage}[t]{0.24\textwidth}
\centering
\begin{tikzpicture}[
scale=0.8,
vertex/.style={circle,draw, minimum size=4mm, inner sep=2pt},
labelnode/.style={draw=none, fill=none, font=\scriptsize}
]

\node[fill=gray!20,vertex] (s) at (0,0) {$s$};
\node[labelnode,text=blue] at (0,0.5) {$1$};

\node[vertex] (a) at (1,0) {$a$};
\node[labelnode,text=blue] at (1,0.5) {$1$};

\node[fill=gray!20,vertex] (b) at (2,1) {$b$};
\node[labelnode,text=red] at (2,1.5) {$1$};

\node[vertex] (t) at (3,0) {$t$};
\node[labelnode,text=blue] at (3,0.5) {$3$};

\node[fill=gray!20,vertex] (c) at (1,-1.5) {$c$};
\node[labelnode,text=red] at (1,-0.9) {$1$};

\draw[->] (s)--(a);
\draw[->] (a)--(b);
\draw[->] (a)--(t);
\draw[->] (b)--(t);
\draw[->] (s)--(c);
\draw[->] (c)--(t);
\draw [->,bend left= 30, blue, dashed] (c) to (s);
\draw [->,bend left= 30, blue, dashed] (b) to (a);

\end{tikzpicture}

\vspace{2mm}
\parbox[c][1.5em][c]{\linewidth}{\centering (2)}
\end{minipage}
\begin{minipage}[t]{0.24\textwidth}
\centering
\begin{tikzpicture}[
scale=0.8,
vertex/.style={circle,draw, minimum size=4mm, inner sep=2pt},
labelnode/.style={draw=none, fill=none, font=\scriptsize}
]

\node[fill=gray!20,vertex] (s) at (0,0) {$s$};
\node[labelnode,text=blue] at (0,0.5) {1};

\node[vertex] (a) at (1,0) {$a$};
\node[labelnode,text=blue] at (1,0.5) {1};

\node[fill=gray!20,vertex] (b) at (2,1) {$b$};
\node[labelnode,text=red] at (2,1.5) {$1$};

\node[vertex] (t) at (3,0) {$t$};
\node[labelnode,text=blue] at (3,0.5) {3};

\node[fill=gray!20,vertex] (c) at (1,-1.5) {$c$};
\node[labelnode,text=blue] at (1,-1) {2};

\draw[->] (s)--(a);
\draw[->] (a)--(b);
\draw[->] (a)--(t);
\draw[->] (b)--(t);
\draw[->] (s)--(c);
\draw[->] (c)--(t);
\draw [->,bend left= 30, blue, dashed] (b) to (a);

\end{tikzpicture}

\vspace{2mm}
\parbox[c][1.5em][c]{\linewidth}{\centering (3)}
\end{minipage}
\begin{minipage}[t]{0.24\textwidth}
\centering
\begin{tikzpicture}[
scale=0.8,
vertex/.style={circle,draw, minimum size=4mm, inner sep=2pt},
labelnode/.style={draw=none, fill=none, font=\scriptsize, text=blue}
]

\node[fill=gray!20,vertex] (s) at (0,0) {$s$};
\node[labelnode] at (0,0.5) {1};

\node[vertex] (a) at (1,0) {$a$};
\node[labelnode] at (1,0.5) {1};

\node[fill=gray!20,vertex] (b) at (2,1) {$b$};
\node[labelnode] at (2,1.5) {2};

\node[vertex] (t) at (3,0) {$t$};
\node[labelnode] at (3,0.5) {3};

\node[fill=gray!20,vertex] (c) at (1,-1.5) {$c$};
\node[labelnode] at (1,-0.9) {2};

\draw[->] (s)--(a);
\draw[->] (a)--(b);
\draw[->] (a)--(t);
\draw[->] (b)--(t);
\draw[->] (s)--(c);
\draw[->] (c)--(t);

\end{tikzpicture}

\vspace{2mm}
\parbox[c][1.5em][c]{\linewidth}{\centering (4)}
\end{minipage}

\caption{Calculation of $prec_A$ (second {\bf forall} loop of Algorithm~\ref{alg:minFoll}); nodes of $A$ are in grey, values of $np_A$ are in red, values of $prec_A$ are in blue,  and values of $ref$ are dashed blue arrows.}
\label{ex:2alg3}
\end{figure}

We now show how the algorithm works on a simple example. 
Consider the net in Figure~\ref{ex:alg3}(1) and let $A=\{s,b,c\}$ (depicted as gray nodes); the initial values of $np_A$, as calculated by the first {\bf forall}, are in red over each node. 
The algorithm then enters into its {\bf while} loop.
In the first iteration, two nodes can be eliminated, $b$ and $c$, and let us assume the algorithm chooses $b$: then, $\mathit{ref}[b]=a$, $np_A(t)=0+1=1$ and $(a,t)$ is not added to $\hat{G}$, since it is already present. This yields the net in Figure~\ref{ex:alg3}(2). 
In the second iteration, the algorithm can choose either $a$ or $c$, say that it chooses $c$; then, $\mathit{ref}[c]=s$, $np_A(t)=1+1=2$ and $(s,t)$ is added to $\hat{G}$, since it is not present. This yields the net in Figure~\ref{ex:alg3}(3). 
Finally, the only node that the algorithm can choose is $a$; then, $\mathit{ref}[a]=s$, $np_A(t)=2+0=2$ and $(s,t)$ is not added to $\hat{G}$. 
This yields the net in Figure~\ref{ex:alg3}(4); having obtained the single edge $(s,t)$, the {\bf while} terminates. 
Then, the algorithm calculates $prec_A$, as shown in Figure~\ref{ex:2alg3}.
Initially, $prec_A[s]=np_A[s]=1$ and $prec_A[t]=|A|=3$, as shown in Figure~\ref{ex:2alg3}(1), where we depict in blue the values of $prec_A$ and in blue dashed arrows the values of $\mathit{ref}$. 
The algorithm then performs a topological visit of $G$. 
In the first iteration of the second {\bf forall}, the algorithm can visit $a$ or $c$; assuming that it chooses $a$, we obtain $prec_A(a)=np_A(a)+prec_A(s)=0+1=1$, see Figure~\ref{ex:2alg3}(2). 
In the second iteration, it can choose $b$ or $c$; assume $c$ and let $prec_A(c)=np_A(c)+prec_A(s)=1+1=2$, see Figure~\ref{ex:2alg3}(3). 
Then, the only choice is $b$ and let $prec_A(b)=np_A(b)+prec_A(a)=1+1=2$, see Figure~\ref{ex:2alg3}(4), where all $prec_A$ are correctly computed. 
In the very last iteration, the algorithm visits $t$ and, being this the first node that has $prec_A$ that equals $|A|$, the algorithm correctly returns $t$.

\subsection{The Main Algorithm}
\label{sec:mainAlgo}
%

Given a non-vulnerable net $(G,s,t)$, Algorithm~\ref{alg:vulnGue} states if $(G\cup \{e\},s,t)$ is vulnerable, where $e$ is the new edge. 
First, we compute $RCC(e)$, while saving all the entries and exits in $\EN$ and $\EX$.
Then, we check if the addition of $e$ is impactless, destructive or impactful.
In the first case, we return {\sc No}: the resulting graph remains non-vulnerable and maintains the very same MIS.
In the second case, we return {\sc Yes}: the new graph becomes vulnerable.
In the third case, we have to run \MS\ on the $\en\ex$-net $RCC(e)$, where $(\en,\ex)$ is the entry-exit pair for $e$. If this call returns {\sc Yes}, we also return {\sc Yes}; otherwise, we return {\sc No} and we use the MIS of $RCC(e)$ returned by \MS\ to build the MIS of $G \cup \{e\}$.

\begin{algorithm}[t]
    \caption{$IncrVuln(G,mis,e)$}
    \begin{algorithmic}
        \State $(rcc,\EN,\EX)=RCC(G,e)$
        \State ($answer$, $\en,\ex$) $=\mathit{TypeOfAdd}(G, mis,\EN,\EX)$
        \If{$answer$ = {\sc Impactless}} \Return ({\sc No}, $G \cup \{e\}$, $mis$)
        \ElsIf{$answer$ = {\sc Destructive}} \Return {\sc Yes}
        \ElsIf{$answer$ = {\sc Impactful}} 
            \If{\MS$(rcc,\en,\ex$) = {\sc Yes}} \Return {\sc Yes}
            \Else \State Let $M$ be the MIS returned by \MS\ on $rcc$ 
            \State \Return ({\sc No}, $G \cup \{e\}$, $mis \cup M$)
            \EndIf
        \EndIf
    \end{algorithmic}
    \label{alg:vulnGue}
\end{algorithm}

\begin{theorem}[Correctness]
\label{thm:correct}
Assuming that $G$ is not vulnerable and that we already know its MIS, $IncrVuln$ states correctly if $G\cup \{e\}$ is vulnerable (answer {\sc Yes}) or not  (answer {\sc No}), for every $e \not\in E_G$; moreover, in case it is not vulnerable, it also correctly calculates its MIS.
\end{theorem}
\begin{proof}
    By Lemma~\ref{redred}, if the addition of $e$ is impactless, then $e$ is redundant; so, the addition of $e$ cannot create any new simple path. Being $G$ not vulnerable, $G\cup \{e\}$ is not vulnerable too and has the same MIS as before.
    
    By Lemma~\ref{vulnvuln}, if the addition of $e$ is destructive, then $G\cup \{e\}$ is vulnerable.
    
    By Lemma~\ref{irrvuln}, if the addition of $e$ is impactful with entry-exit pair $(\en,\ex)$, then $G\cup \{e\}$ is vulnerable if and only if the $\en\ex$-net $RCC(e)$ is vulnerable. In case it is not, the MIS of the new graph is exactly the old MIS extended with the MIS of $RCC(e)$ (see Corollary~\ref{cor:newMIS}).
\end{proof}

\label{subsec:Amortized Analysis}

\begin{proposition}
\label{prop:cost}
Let $G=(V,E)$ be a non-vulnerable graph with $|E|=m$.
    An impactless addition costs $O(m)$.
    An impactful addition of an edge $e$ costs $O(m+kl)$, where $l$ is the number of edges of $RCC(e)$ and $k$ is the number of edges in $RCC(e)$ that become irredundant after adding $e$ (including $e$ itself).
\end{proposition}
\begin{proof}
In the case of an impactless addition, the only costs are those of running $RCC$ and $\mathit{TypeOfAdd}$ (i.e., Algorithms~\ref{alg:RCC} and~\ref{alg:R-I-V}), and both are $O(m)$ (see Proposition~\ref{prop:costRCC} and Lemma~\ref{prop_cost-R-I-V}).
In the case of an impactful addition, there is also an extra cost due to the execution of \MS\ on $RCC(e)$, whose complexity is given by Proposition~\ref{prop:MS}. 
\end{proof}

\begin{theorem}[Amortized analysis]
\label{thm:anal}
Assume to start with an empty graph and to have a sequence of $m$ non-destructive edge additions; the amortized cost of every addition is $O(m)$.
\end{theorem}

\begin{proof}
    Let $r$ be the number of edges that, when added, are redundant (i.e., the number of impactless additions) and let $d$ be the number of edges that, when added, are irredundant (i.e., the number of impactful additions); thus, $m = r + d$.
    Furthermore, for each of the $d$ edges $e_1,\ldots,e_d$, let $h_i$ be the number of edges that become irredundant after adding edge $e_i$, including itself;
thus, $\sum_{i=1}^d h_i\le r+d$. By Proposition~\ref{prop:cost}, each of the $r$ impactless additions costs $O(m)$ and the impactful addition of edge $e_i$ costs $O(m+h_im)$.
    Thus, the overal cost of adding all the $r+d$ edges is
    $$C\in O\left(rm+dm+m\cdot\sum_{i=1}^d h_i\right)=O((r+d)m+m(r+d))=O((r+d)m)$$
    Hence, the amortized cost is $\frac{C}{r+d}\in O(m)$.
\end{proof}
We remark that, in the previous Theorem, we assumed to start with the empty graph and build the net edge-by-edge; this helped us in carrying out the amortized analysis. However, Algorithm~\ref{alg:vulnGue} works correctly also when it receives in input a non-vulnerable non-empty starting graph (together with its MIS).


\section{Conclusions}

Vulnerability to the Braess Paradox is a problem that can arise in nets which can change dynamically, gaining or losing links between their nodes. Since nets can be very big, it is important to determine whether edge additions or deletions change the vulnerability of the net without running every time the static $O(m^2)$ algorithm by \cite{MS23} on the whole net.
In the presented approach, we study the addition of a single edge and classify it into three possible scenarios: impactful, impactless, or destructive. Our algorithm recognizes in linear time the type of addition and runs the algorithm of \cite{MS23} to solve the impactful case. Since we avoid checking the same subnet twice, our strategy achieves an amortized cost of $O(m)$ for every addition; thus, we are paying $O(m^2)$ to add $m$ edges.

Our algorithm works only in the incremental case (i.e., when adding edges).
In a companion paper, we consider the decremental case (i.e., when removing edges) and provide also in that scenario an amortized linear time algorithm,  by leveraging the approach in \cite{CGS19}.
Moreover, we also prove \cite{F26} that these two dynamic algorithms can be extended to multi-source/multi-target nets (called \emph{multi-commodities}), as done in the static case in \cite{FGS23}.
In contrast, the fully dynamic case (i.e., when edge additions and removals can arbitrarily interleave) seems to be a more complex problem, since 
it is not clear what information should be kept from one addition/deletion to the next one to enable an efficient vulnerability checking. 
This is a challenging line for future research.

Another challenging line is the problem of finding the minimal set of edges whose removal turns a vulnerable net into a non-vulnerable one.
We conjecture that this problem is NP-hard and, consequently, it would be natural to find approximation algorithms for this optimization problem.

Finally, we also plan to study dynamic algorithms when $G$ is an undirected graph. In that case, absence of Braess's Paradox coincides with having a series-parallel MIS \cite{Milch06}; differently from the directed case, this can be checked in linear time. Thus, it would be interesting to see if an amortized sub-linear time incremental/decremental/fully-dynamic algorithm for vulnerability exists.

\bibliography{threeForms,chan}

\clearpage


\appendix

\section{Omitted Material from Section 2}

\subsection{Traffic Networks, Braess Paradox and Vulnerability}
\label{app:braess}

A \emph{flow} for a net $(G,s,t)$ is a function 
$\varphi:\SP(G) \rightarrow \mathbb{R}^+$.
The value of a flow is the sum of the values sent over all paths. 
A flow induces a unique flow $\varphi(e)$ on edges: for any edge $e\in E$, 
$\varphi(e)=\sum_{P\in \SP(G): e\in P}\varphi(P)$. 

A {\em latency function} $l_e:\mathbb{R}^+\rightarrow \mathbb{R}^+$ 
assigns to each edge $e$ a latency as a function of the flow in it;
as usual, we only consider continuous and non-decreasing latency functions.

The latency of an $st$-path $P$ under a flow $\varphi$ is the sum of the latencies of 
all edges in the path under $\varphi$, i.e., $l_P(\varphi)=\sum_{e\in P} l_e(\varphi(e))$.
If $H$ is a subgraph of $G$, we denote with 
$l|_H$ the restriction of the latency function $l$ on the edges of $H$.

Given a net $G$, a real number $r\in\mathbb{R}^+$ and a latency function $l$, 
we call the triple $(G, r, l)$ an {\em instance} (of $G$). A flow $\varphi$ is {\em feasible} for 
$(G, r, l)$ if the value of $\varphi$ is $r$. Notice that, since we do not have any
constraint on edges or vertices, every $r$ admits at least one feasible flow.

A feasible flow $\varphi$ for $(G, r, l)$ is at {\em Wardrop equilibrium} 
(or is a {\em Wardrop flow})\footnote{
	Here we use as definition the characterisation of Wardrop flow
	put forward by Proposition 2.2 in \cite{Rough06}.
} if, for all $st$-paths $P$ and $Q$  
such that $\varphi(P)>0$, we have $l_P(\varphi)\leq l_Q(\varphi)$.
In particular, this implies that, if $\varphi$ is a Wardrop flow, 
all $st$-paths to which $\varphi$ assigns a positive flow
have the same latency. It is known \cite{Rough06} that every
instance admits a Wardrop flow and that different Wardrop flows for the same
instance have the same latency along all $st$-paths with a positive flow. Thus,
we denote with $L(G, r, l)$ the latency of all $st$-paths 
with positive flow at Wardrop equilibrium; whenever $r = 0$,
we let $L(G, r, l)$ be 0.

{\em Braess paradox} \cite{braessFormal,braessOriginal} originates when latency 
at Wardrop equilibrium decreases because of removing some edges (or equivalently, by increasing 
latency functions on them): an instance $(G, r, l)$ suffers from Braess paradox 
if a subgraph of $G$ exists with a lower latency.
In \cite{Rough06} it is shown that 
it is \NP-hard to decide whether or not an instance of a net  
suffers from the Braess paradox. A question raised in \cite{Rough06} is to study
the problem from a graph-theoretical perspective, 
that is whether a net admits instances suffering from 
the paradox.
This leads to the following definition \cite{Rough06}.
\begin{definition}
	A net $G$ is {\em vulnerable} if there exist a value $r$, a latency function
	$l$ and a subgraph $H$ of $G$ such that $L(G, r, l)>L(H, r, l|_H)$.
\end{definition}

\subsection{The \MS\ Algorithm and its Refined Complexity Analysis}

We report in Figure~\ref{alg:algMS} the algorithm \MS\ as devised in \cite{MS23}.

\begin{algorithm}
\caption{\textsc{SP\_Test}($G, s, t$)}
\label{alg:algMS}
\begin{algorithmic}[1]
\State \textbf{Input:} A directed two-terminal graph $G$ with source $s$ and sink $t$.
\State \textbf{Output:} ``Yes'' if $MIS(G)$, call it $\tilde{G}$, is series-parallel, ``No'' otherwise.

\If{$\tilde{G}$ is a single edge $(s,t)$}
    \State \textbf{return} ``Yes''
\EndIf

\State Perform \textsc{Series\_Decomposition}$(G, s, t)$.

\ForAll{$s't'$-subgraph $G'$ returned by \textsc{Series\_Decomposition}}
    
    \State Perform \textsc{Fast\_Parallel\_Decomposition}$(G', s', t')$.
    
    \If{it returns ``No'' or a single subgraph whose MIS is not a single edge $(s',t')$}
        \State \textbf{return} ``No''
    \EndIf

    \ForAll{$s''t''$-subgraph $G''$ returned by \textsc{Fast\_Parallel\_Decomposition}}
        \State Perform \textsc{SP\_Test}$(G'', s'', t'')$ recursively. \label{line:recCall}
    \EndFor

\EndFor

\If{all recursive calls in line~\ref{line:recCall} return ``Yes''}
    \State \textbf{return} ``Yes''
\Else
    \State \textbf{return} ``No''
\EndIf

\end{algorithmic}
\end{algorithm}
\begin{algorithm}
\caption{\textsc{Series\_Decomposition}($G, s, t$)}
\begin{algorithmic}[1]
\State \textbf{Input:} A directed two-terminal graph $G$ with source $s$ and sink $t$.
\State \textbf{Output:} The maximum number $k$ of two-terminal subgraphs $G_1, \dots, G_k$ of $G$, such that $\tilde{G}$ is obtained by series composition of $\tilde{G}_1, \dots, \tilde{G}_k$, and that for each $1 \le i < k$, $G_i$ and $\bigcup_{j>i} G_j$ share $v_{i+1}$ only.

\State Find all $st$-articulation points $v_0 = s, v_1, \dots, v_{k-1}, v_k = t$ appearing in this order on an $st$-path.

\For{$i = 1$ to $k$}
    \State Find a set $V_i$ of vertices $x$ such that there exists a $v_{i-1}x$-path in $G$ consisting of edges neither leaving $v_i$ nor entering a vertex in $\bigcup_{j=1}^{i-1} V_j$.
    \State Return the graph induced by $V_i$ as $G_i$.
\EndFor

\end{algorithmic}
\end{algorithm}

\begin{algorithm}
\caption{\textsc{Fast\_Parallel\_Decomposition}($G, s, t$)}
\begin{algorithmic}[1]
\State \textbf{Input:} A directed two-terminal graph $G$ with source $s$ and sink $t$.
\State \textbf{Output:} Either the maximum number $k$ of two-terminal subgraphs $G_1, \dots, G_k$ of $G$, sharing $s$ and $t$ only, such that $\tilde{G}$ is obtained by parallel composition of $\tilde{G}_1, \dots, \tilde{G}_k$, or ``No'' meaning $\tilde{G}$ is not series-parallel.

\State Set $i = 1$ and suppose that there are $d$ edges leaving $s$, denoted by $(s,u_1), \dots, (s,u_d)$.

\For{$j = 1$ to $d$}
    \If{$u_j = t$}
        \State Define $P_i$ as the $st$-path consisting of the single edge $(s,u_j)$.
        \State Increment $i$ by $1$.
    \Else
        \State Perform DFS starting at $s$, traversing the edge $(s,u_j)$ first, and avoiding vertices visited by previous DFS for smaller $j$.
        \If{a vertex incident to an edge entering $t$ is reached}
            \State Define $P_i$ as the $st$-path visited by the DFS.
            \State Quit the DFS and increment $i$ by $1$.
        \Else
            \State If DFS backtracks to $s$ with no $st$-path found, quit the DFS.
        \EndIf
    \EndIf
\EndFor

\State Suppose that we have $\ell$ $st$-paths $P_1, \dots, P_\ell$.

\For{$i = 1$ to $\ell$}
    \State Let $U_i$ be the set of internal vertices of $P_i$.
    \State Let $U_i^{o}$ be the set of vertices not in $P_i$ and incident to an edge leaving a vertex in $U_i$.
    \State Let $U_i^{i}$ be the set of vertices not in $P_i$ and incident to an edge entering a vertex in $U_i$.
\EndFor

\State Find the set $S$ of vertices $x$, excluding $s$, such that there exists an $sx$-path disjoint with $P_1, \dots, P_\ell$, by DFS starting at $s$ and avoiding vertices in $\bigcup_{i=1}^{\ell} U_i \cup \{t\}$.

\State Find the set $T$ of vertices $x$, excluding $t$, such that there exists an $xt$-path disjoint with $P_1, \dots, P_\ell$, by reverse DFS starting at $t$ and avoiding vertices in $\bigcup_{i=1}^{\ell} U_i \cup \{s\}$.

\For{$i = 1$ to $\ell$}

    \State Find the set $O_i$ of vertices visited by DFS starting at every vertex $v \in U_i^{o}$, avoiding vertices in $P_i$, in $\bigcup_{j<i} O_j$, and already in $O_i$.
    
    \State During the DFS:
    \begin{itemize}
        \item If a vertex in $U_j$ with $j \ne i$ is reached, return ``No''.
        \item If a vertex incident to an edge entering a vertex in $O_j \cap T$ for some $j < i$ is reached, return ``No''.
    \end{itemize}

    \State Find the set $I_i$ of vertices visited by reverse DFS starting at every vertex $v \in U_i^{i}$, avoiding vertices in $P_i$, in $\bigcup_{j<i} I_j$, and already in $I_i$.

    \State During the reverse DFS:
    \begin{itemize}
        \item If a vertex in $U_j$ with $j \ne i$ is reached, return ``No''.
        \item If a vertex incident to an edge leaving a vertex in $I_j \cap S$ for some $j < i$ is reached, return ``No''.
    \end{itemize}

\EndFor

\For{$i = 1$ to $\ell$}
    \State Define $V_i = \{s,t\} \cup U_i \cup (O_i \cap I_i) \cup (O_i \cap T) \cup (I_i \cap S)$.
    \State Return the graph induced by $V_i$ as $G_i$.
\EndFor

\end{algorithmic}
\end{algorithm}
The authors proved a worst-case complexity of $O(m^2)$; however, we can refine this analysis (without changing the worst-case scenario) as stated in Proposition~\ref{prop:MS}, that we now prove.

\begin{proof}[\bf Proof of Proposition~\ref{prop:MS}]
    The \MS\ algorithm is composed by a sequence of recursive calls that alternate two functions: \textit{Series\_Decomposition} (SD) and \textit{Fast\_Parallel\_Decomposition} (FPD), each costing $O(m)$. 
By abuse of notation, we call $s$ and $t$ the nodes source and target in each decomposed subgraph.

    SD first runs on the entire graph; 
    it finds a certain number $l$ of articulation points (nodes that, once removed, disconnect $s$ from $t$). So, all $st$-paths pass through these articulation points. The algorithm returns $l$ subgraphs, which are node-disjoint except for the articulation points; these represent the series components of $G$. For the nature of the articulation points, in each of those $l$ graphs there must be at least one edge that belongs to a simple path, which hence is irredundant. 
    
    FPD now runs on each of the series component made by SD. It first finds a number $h$ which represents the maximal number of simple paths which are node-disjoint, except for $\{s,t\}$. If $h=1$, either the component is the single edge $(s,t)$ or the entire $G$ is vulnerable. If $h>1$, the algorithm can either find that the graph is vulnerable or it returns $h$ graphs, node-disjoint except for $\{s,t\}$, representing the parallel components of $G$. Since $h$ is the maximal number of node-disjoint simple paths, each of those is composed by at least one distinct irredundant edge.
    
    It is then possible to build a recursion tree in which:
    \begin{itemize}
        \item Each level has an overall cost of $O(m)$ because each series and parallel component is node-disjoint from the others, except for one or two nodes.
        \item The depth of the tree is $O(k)$ because in each iteration FPD isolates at least one irredundant edge from the others.
    \end{itemize}
Consequently, the cost of the \MS\ algorithm is $O(km)$.
\end{proof}


\section{Proofs for Section 3}

\begin{proposition}
\label{prop:leq}
If $G$ is an $st$-connected DAG, then
\begin{enumerate}
\item If $a<b$, then it cannot be $b<a$;
\item $a < b$ if and only if there exists a simple path from $a$ to $b$ in $G$;
\item ‘$<$' is transitive.
\end{enumerate}
\end{proposition}
\begin{proof}
The first claim holds by acyclicity of $G$.

For the second claim, the ‘only if' is trivial. For the other direction, assume a simple path $P: a \leadsto b$; since the graph is $st$-connected, there exist two other simple paths $P' : s \leadsto a$ and $P'' : b \leadsto t$. Now consider $s \stackrel {P'} \leadsto a \stackrel {P} \leadsto b \stackrel {P''} \leadsto t$: if it was not simple, then $G$ would have a cycle; so, $a < b$.

For the last claim, assume that $a<b$ and $b<c$; by the previous claim of this Proposition, there exist two simple paths $P : a \rightsquigarrow b$ and $P' : b \rightsquigarrow c$. Now consider the path $a \stackrel P \rightsquigarrow b \stackrel {P'} \rightsquigarrow c$: if it was not simple, then $G$ would have a cycle; thus, $a<c$. 
\end{proof}

\begin{proposition}\label{greaterredundant}
    If $u,v \in MIS(G)$ and $v<u$, then there cannot exist two simple node-disjoint paths $P_{su}:s\rightsquigarrow u$ and $P_{vt}:v\rightsquigarrow t$.
\end{proposition}
\begin{proof}
    By contradiction. Since $v < u$, let $P:s\rightsquigarrow v \rightsquigarrow u \rightsquigarrow t$. 
    Since $P$ and $P_{su}$ both pass through $s$ and $u$, but $P_{su}$ does not contain $v$,
    it must be that $P$ and $P_{su}$ must split before $v$ and then join between $u$ and $v$, viz.
     $P=s\rightsquigarrow a \rightsquigarrow v \rightsquigarrow b \rightsquigarrow u \rightsquigarrow t$ and $P_{su}=s \rightsquigarrow a \rightsquigarrow b \rightsquigarrow u$, where $a$ is the node where they last split and $b$ is the node where they first join afterwards.
     Hence, $a \stackrel P \rightsquigarrow b\ \cap\ a \stackrel {P_{su}} \rightsquigarrow b = \{a,b\}$.
    A similar reasoning can be done for $P$ and $P_{vt}$, where we let $c$ be the node where they last split and $d$ the node where they first join afterwards. 
We can pictorially represent this as follows (where $P$ is in black, $P_{su}$ is in green, and $P_{vt}$ is in red):
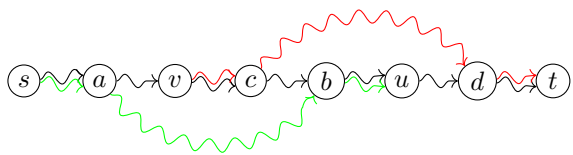
\[
\begin{tikzpicture}[every node/.style={circle,draw, minimum size=4mm, inner sep=2pt}]
\node (s) at (0,0) {$s$};
\node (a) at (1,0) {$a$};
\node (v) at (2,0) {$v$};
\node (c) at (3,0) {$c$};
\node (b) at (4,0) {$b$};
\node (u) at (5,0) {$u$};
\node (d) at (6,0) {$d$};
\node (t) at (7,0) {$t$};
\draw [->, decorate, decoration={snake}] ([yshift=10pt]s) -- ([yshift=10pt]a);
\draw [->, decorate, decoration={snake}, green] ([yshift=-10pt]s) -- ([yshift=-10pt]a);
\draw [->, decorate, decoration={snake}] (a) -- (v);
\draw [->, decorate, decoration={snake}] ([yshift=-10pt]v) -- ([yshift=-10pt]c);
\draw [->, decorate, decoration={snake}, red] ([yshift=10pt]v) -- ([yshift=10pt]c);
\draw [->, decorate, decoration={snake}] (c) -- (b);
\draw [->, decorate, decoration={snake}] ([yshift=10pt]b) -- ([yshift=10pt]u);
\draw [->, decorate, decoration={snake}, green] ([yshift=-10pt]b) -- ([yshift=-10pt]u);
\draw [->, decorate, decoration={snake}] (u) -- (d);
\draw [->, decorate, decoration={snake}] ([yshift=-10pt]d) -- ([yshift=-10pt]t);
\draw [->, decorate, decoration={snake}, red] ([yshift=10pt]d) -- ([yshift=10pt]t);
\draw [->, bend right= 50, decorate, decoration={snake},green] (a) to (b);
\draw [->, bend left= 50, decorate, decoration={snake},red] (c) to (d);
\end{tikzpicture}
\]

Now consider the three paths: 
$P_1=a\rightsquigarrow v \rightsquigarrow c {\ \color{red}\rightsquigarrow\ } d$, 
$P_2= a {\ \color{green}\rightsquigarrow\ } b \rightsquigarrow u \rightsquigarrow d$, and 
$P_3=c \rightsquigarrow b$.
If $a {\ \color{green}\rightsquigarrow\ } b\ \cap\ c {\ \color{red}\rightsquigarrow\ } d = \emptyset$, then we would have found a homeomorphic copy of $W$ (with endpoints $a$ and $b$, and with the bridge being $P_3$), and this would make $MIS(G)$ not series-parallel: contradiction.
Otherwise, let $w$ be the last intersection along $a {\ \color{green}\rightsquigarrow\ } b$ (i.e., $w {\ \color{green}\rightsquigarrow\ } b\ \cap\ c {\ \color{red}\rightsquigarrow\ } d = \emptyset$); again we could find a homeomorphic copy of $W$ (this time with endpoints $c$ and $d$, and with the bridge being $w {\ \color{green}\rightsquigarrow\ } b$).
\end{proof}

\begin{proof}[\bf Proof of Lemma~\ref{redred}]
    If the addition of $e$ is impactless, then every exit precedes every entry in $G$. By contradiction, assume that $e$ is irredundant in $G \cup \{e\}$; so, there must exist a simple $st$-path $P$ in $G\cup \{e\}$ that uses $e$. Let $\en$ be the last node of $P$ in $MIS(G)$ before an edge in $RCC(e)$ and $\ex$ be the first node of $P$ in $MIS(G)$ after an edge in $RCC(e)$. Because $P$ is simple and contains $e$, it must be that $P=s\rightsquigarrow \en \rightsquigarrow \ex \rightsquigarrow t$ with $\en \neq \ex$. 
    This would imply the existence of two simple node-disjoint paths $P_{s\en}:s\rightsquigarrow \en$ and $P_{\ex t}: \ex \rightsquigarrow t$. Since $\ex < \en$, this would contradict Proposition~\ref{greaterredundant}. Since no new paths are created with the add of $e$, $MIS(G)=MIS(G\cup\{e\})$.
\end{proof}

\begin{proof}[\bf Proof of Lemma~\ref{irredirred}]
Let $\en$ and $\ex$ be the unique linkable pair such that $\en < \ex$; hence, there is a simple path $\en \rightsquigarrow \ex$ in $G$. Moreover, in $G$ there exists also a path $P=s\rightsquigarrow \en \rightsquigarrow \ex \rightsquigarrow t$ that must be simple, otherwise $MIS(G)$ would be cyclic. Then, there exist two simple node-disjoint paths $P_{s\en}=s\stackrel P \rightsquigarrow \en$ and $P_{\ex t}=\ex\stackrel P \rightsquigarrow t$ both in $MIS(G)$. By definition of $RCC(e)$, there exists therein a simple path $P'=\en \rightsquigarrow \ex$ passing through $e$. Then, the path $P_{s\en} P' P_{\ex t}$ in $G \cup\{e\}$ is simple and contains $e$; hence, $e$ is irredundant in the new graph.

All the new simple $st$-paths that touch $RCC(e)$ must enter $RCC(e)$ in $\en$ and exit $RCC(e)$ in $\ex$. Indeed, if there was a path (containing $e$) from some $\en'$ to some $\ex'$ different (as a pair) from $\en$ and $\ex$, then $\ex' < \en'$ in $G$, being the addition impactful; we would then obtain a contradiction by a reasoning similar to that in the proof of Proposition~\ref{redred}.
Hence, the irredundant part of $RCC(e)$ has a unique source $\en$ and target $\ex$; since $\en$ and $\ex$ are linkable, the thesis follows.
\end{proof}

\begin{proof}[\bf Proof of Lemma~\ref{vulnvuln}]
    If the addition of $e$ is destructive, it means that it is not impactless nor impactful; so, there are different possibilities:
    \begin{enumerate}
        \item If there is an entry $\en$ and an exit $\ex$ that are not comparable, then there is a new simple path $\en \rightsquigarrow \ex$ which contains $e$ and can be used to build a homeomorphic copy of $W$, together with the paths $s \rightsquigarrow \en \rightsquigarrow t$ (that does not contain $\ex$) and $s \rightsquigarrow \ex \rightsquigarrow t$ (that does not contain $\en$).

        \item If there is only one entry $\en$ and exit $\ex$ such that $\en < \ex$, but $\en$ and $\ex$ are not linkable, the thesis follows by Definition~\ref{def:linkable} and by \cite[Theorems 1 and 3]{Duf65}. 

        \item If there are two different linkable pairs $(\en_1,\ex_1)$ and $(\en_2,\ex_2)$ such that $\en_i < \ex_i$ (for $i =1,2$), we can assume that $\en_1 \neq \en_2$ (that case for $\ex_1 \neq \ex_2$ is dual). Then, there must exist two simple paths in the new graph  $P_1 = s\rightsquigarrow \en_1 \rightsquigarrow \ex_1 \rightsquigarrow t$ and $P_2 = s \rightsquigarrow \en_2 \rightsquigarrow \ex_2 \rightsquigarrow t$ such that both $\en_1 \stackrel{P_1}\rightsquigarrow \ex_1$ and $\en_2 \stackrel{P_1}\rightsquigarrow \ex_2$ contain $e = (u,v)$.
        \begin{enumerate}
            \item  If $\en_1$ and $\en_2$ are comparable (w.l.o.g., suppose $\en_1< \en_2$), there must exist a simple path $s \leadsto \en_1\rightsquigarrow \en_2 \leadsto t$  in $MIS(G)$. 
Since $P_1$ and $P_2$ both contain $e$, they must have at least one node in common, call $c$ the first. 
Now, $\en_1 \stackrel{P_1}\rightsquigarrow c$ and $\en_2\stackrel{P_2}\rightsquigarrow c$ both belong to $RCC(e)$ and are disjoint by the way in which we chose $c$.
Hence, there must exist an exit node $\ex$ such that $c\rightsquigarrow u \rightarrow v \rightsquigarrow \ex$ is in $RCC(e)$ and is disjoint from $\en_1\rightsquigarrow c$ and $\en_2\rightsquigarrow c$ because of acyclicity. 
So, the two parallel node-disjoint simple paths $\en_1\rightsquigarrow \en_2\rightsquigarrow \ex$ in $MIS(G)$ and $\en_1\rightsquigarrow c\rightsquigarrow u \rightarrow v \rightsquigarrow \ex$ in $RCC(e)$ are connected by $\en_2\rightsquigarrow c$; this yields an homeomorphic copy of $W$.

            \item  If $\en_1$ and $\en_2$ are not comparable, they belong to different parallel components; so, there exists the last articulation point  $a$ such that $s \le a < \en_1$ and $s \le a < \en_2$. Let $c$ and $\ex$ be the nodes like in point 3(a); then, $a \rightsquigarrow \en_2\rightsquigarrow \ex$ in $MIS(G)$ is disjoint from $a\rightsquigarrow \en_1\rightsquigarrow c \rightsquigarrow u \rightarrow v \rightsquigarrow \ex$ because $a\rightsquigarrow \en_1$ is on a different parallel component and $\en_1\rightsquigarrow c\rightsquigarrow u \rightarrow v \rightsquigarrow \ex$ is in $RCC(e)$. Moreover, $\en_2\rightsquigarrow c$ is disjoint from the first path because it is in $RCC(e)$ and from the second one for the same reasons as in case 3(a). This again suffices to conclude.
        \end{enumerate}
        \item Finally, if there is only one linkable pair with the correct ordering (with the entry that proceeds the exit), but there exists another entry-exit pair whose entry does not follow the exit, then either that latter pair is uncomparable (and we fall in case 1 of this proof) or the entry proceeds the exit but they are not linkable (and we fall in case 2 of this proof).
        \qedhere
    \end{enumerate}
\end{proof}

\newcommand{\Pleads}{{\color{blue}\, \leadsto\, }}
\newcommand{\PPleads}{{\color{red}\, \leadsto\, }}
\begin{proof}[\bf Proof of Proposition~\ref{Prop:RE}]
If the addition of $e$ is impactless, the claim is trivial since the addition of $e$ does not change the set of simple $st$-paths.
The delicate case is when the addition is impactful, and let $(\en,\ex)$ be the entry-exit pair for $e$.
By contradiction, assume that there exists a non-empty set $E' \subseteq E_G$ such that its edges are redundant in $G$, do not belong to $RCC(e)$, but are irredundant in $G \cup \{e\}$. This means that, in the latter graph, for every edge in $E'$ there must exist a simple $st$-path $P$ that passes through it. Let us fix any such $P$. First, $P$ cannot be simple in $G$, and so it must touch $RCC(e)$, by entering this set in $\en$ and exiting it in $\ex$. Hence, 
$$P = s \Pleads \en \Pleads \ex \Pleads t$$ 
where, $\en \Pleads \ex$ is fully contained in $RCC(e)$. 
To simplify reading, we shall always denote portions of $P$ in blue.
Moreover, since $\en < \ex$ in $G$, there exists a simple path $P'$ in $G$ such that
$$P' = s \PPleads \en \PPleads \ex \PPleads t$$ 
To simplify reading, we shall always denote portions of $P'$ in red.

Since $P$ touches $E'$ at least once and $E'\cap RCC(e) = \emptyset$, there must be edges of $E'$ in $s \Pleads \en$ and/or in $\ex \Pleads t$. Let $e'$ be the first edge of $E'$ in $s \Pleads \en$ (if any) and $e''$ be the last edge of $E'$ in $\ex \Pleads t$ (if any). Then:
\begin{itemize}
\item If there is such an $e'$ in $s \Pleads \en$, then the path $s \Pleads \en \PPleads t$ cannot be simple, otherwise $e'$ would not be redundant in $G$. Hence, $s \Pleads \en$ and $\en \PPleads t$ must share at least another vertex, apart from $\en$; let $\hat u$ be their first intersection along $s \Pleads \en$. Now, $e'$ must follow $\hat u$ in $s \Pleads \en$, otherwise the path $s \Pleads \hat u \PPleads t$ would be simple in $G$ and contain $e'$. Being $e'$ the first vertex of $E'$ in $s \Pleads \en$, then $s \Pleads \hat u$ is fully contained in $MIS(G)$.

If there is no such an $e'$ in $s \Pleads \en$, then we let $\hat u = \en$ in what follows.

\item Dually, if there is such an $e''$ in $\ex \Pleads t$, then the path $s \PPleads \ex \Pleads t$ cannot be simple, otherwise $e''$ would not be redundant in $G$. Hence, $s \PPleads \ex$ and $\ex \Pleads t$ must share at least another vertex, apart from $\ex$; let $\hat v$ be their last intersection along $\ex \Pleads t$. Now, $e''$ must proceed $\hat v$ in $\ex \Pleads t$, otherwise the path $s \PPleads \hat v \Pleads t$ would be simple in $G$ and contain $e''$. Being $e''$ the last vertex of $E'$ in $\ex \Pleads t$, then $\hat v \Pleads t$ is fully contained in $MIS(G)$.

If there is no such an $e''$ in $\ex \Pleads t$, then we let $\hat v = \ex$ in what follows.
\end{itemize}

We now consider three cases:
\begin{enumerate}

\item If $\hat v \neq \ex$ and $\hat v \in \en \PPleads \ex$, then we could contradict linkability of $\en$ and $\ex$.
Consider
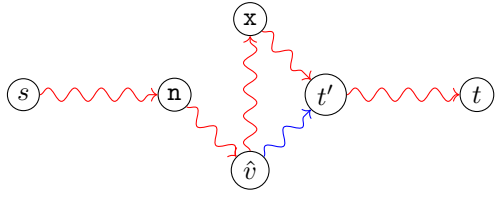
\[
\begin{tikzpicture}[every node/.style={circle,draw, minimum size=4mm, inner sep=2pt}]
\node (s) at (0,0) {$s$};
\node (n) at (2,0) {$\en$};
\node (x) at (3,1) {$\ex$};
\node (v) at (3,-1) {$\hat v$};
\node (t') at (4,0) {$t'$};
\node (t) at (6,0) {$t$};
\draw [->, decorate, decoration={snake}, red] (s) to (n);
\draw [->, decorate, decoration={snake}, red] (n) -- (v);
\draw [->, decorate, decoration={snake}, red] (v) -- (x);
\draw [->, decorate, decoration={snake}, red] (x) -- (t');
\draw [->, decorate, decoration={snake}, red] (t') to (t);
\draw [->, decorate, decoration={snake}, blue] (v) -- (t');
\end{tikzpicture}
\]
where $t'$ is the first intersection along $\hat v \Pleads t$ with $\ex \PPleads t$ (at worst, $t'=t$).
If we add the edge $(\en,\ex)$, we would have an embedding of $W$. Indeed, 
being $\hat v$ the last intersection along $\ex \Pleads t$ with $s \PPleads \ex$, we have that 
$\hat v \Pleads t' \ \cap\ \hat v \PPleads t' = \{\hat v, t'\}$; moreover, 
$\hat v \Pleads t' \ \cap\ s \PPleads \hat v = \{\hat v\}$ and
$\hat v \Pleads t' \ \cap\ t' \PPleads t = \{t'\}$, otherwise $MIS(G)$ would contain a cycle.

\item If $\hat u \neq \en$ and $\hat u \in \en \PPleads \ex$, again we could contradict linkability of $\en$ and $\ex$.
Consider
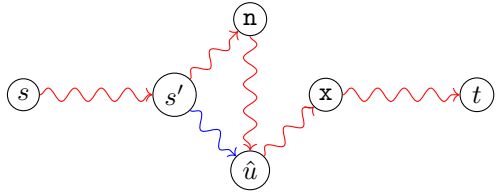
\[
\begin{tikzpicture}[every node/.style={circle,draw, minimum size=4mm, inner sep=2pt}]
\node (s) at (0,0) {$s$};
\node (s') at (2,0) {$s'$};
\node (n) at (3,1) {$\en$};
\node (u) at (3,-1) {$\hat u$};
\node (x) at (4,0) {$\ex$};
\node (t) at (6,0) {$t$};
\draw [->, decorate, decoration={snake}, red] (s) to (s');
\draw [->, decorate, decoration={snake}, red] (s') -- (n);
\draw [->, decorate, decoration={snake}, blue] (s') -- (u);
\draw [->, decorate, decoration={snake}, red] (n) -- (u);
\draw [->, decorate, decoration={snake}, red] (u) to (x);
\draw [->, decorate, decoration={snake}, red] (x) -- (t);
\end{tikzpicture}
\]
where $s'$ is the last intersection along $s \Pleads \hat u$ with $s \PPleads \en$ (at worst, $s'=s$).
If we add the edge $(\en,\ex)$, we would have an embedding of $W$. Indeed, 
being $\hat u$ the first intersection along $s \Pleads \en$ with $\en \PPleads t$, we have that 
$s' \Pleads \hat u \ \cap\ s' \PPleads \hat u = \{s',\hat u\}$; moreover, 
$s' \Pleads \hat u \ \cap\ \hat u \PPleads t = \{\hat u\}$ and
$s' \Pleads \hat u \ \cap\ s \PPleads s' = \{s'\}$, otherwise $MIS(G)$ would contain a cycle.

\item Otherwise, we have the following situation:
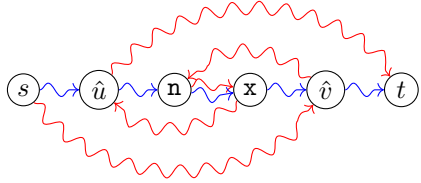
\vspace*{-.3cm}
\[
\begin{tikzpicture}[every node/.style={circle,draw, minimum size=4mm, inner sep=2pt}]
\node (s) at (0,0) {$s$};
\node (u) at (1,0) {$\hat u$};
\node (n) at (2,0) {$\en$};
\node (x) at (3,0) {$\ex$};
\node (v) at (4,0) {$\hat v$};
\node (t) at (5,0) {$t$};
\draw [->, bend right= 50, decorate, decoration={snake}, red] (s) to (v);
\draw [->, bend right= 40, decorate, decoration={snake}, red] (v) to (n);
\draw [->, bend left= 10, decorate, decoration={snake}, red] (n) to (x);
\draw [->, bend left= 40, decorate, decoration={snake}, red] (x) to (u);
\draw [->, bend left= 50, decorate, decoration={snake}, red] (u) to (t);
\draw [->, decorate, decoration={snake}, blue] (s) to (u);
\draw [->, decorate, decoration={snake}, blue] (u) to (n);
\draw [->, bend right= 10, decorate, decoration={snake}, blue] (n) to (x);
\draw [->, decorate, decoration={snake}, blue] (x) to (v);
\draw [->, decorate, decoration={snake}, blue] (v) to (t);
\end{tikzpicture}
\vspace*{-.3cm}
\]
We now let
\begin{itemize}
\item $s'$ be the last intersection along $s \PPleads \hat v$ with $s \Pleads \hat u$ (at worst, $s' = s$);
\item $t'$ be the first intersection along $\hat u \PPleads t$ with $\hat v \Pleads t$ (at worst, $t' = t$);
\item $\bar v$ be the last intersection along $\hat v \PPleads \hat u$ with $\hat v \Pleads t'$ (at worst, $\bar v = \hat v$);
\item $\bar u$ be the first intersection along $\bar v \PPleads \hat u$ with $s' \Pleads \hat u$ (at worst, $\bar u = \hat u$).
\end{itemize}
We now show a contradiction with non-vulnerability of $G$ (that we silently assume whenever we consider an edge addition); to this aim, let us consider the following subgraph of $G$ (actually, of $MIS(G)$):
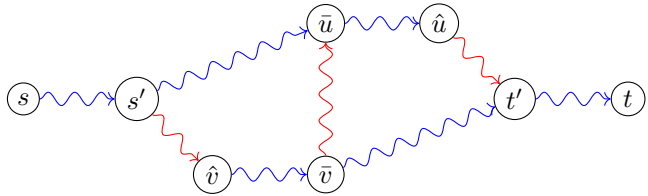
\[
\begin{tikzpicture}[every node/.style={circle,draw, minimum size=4mm, inner sep=2pt}]
\node (s) at (-.5,0) {$s$};
\node (s') at (1,0) {$s'$};
\node (bu) at (3.5,1) {$\bar u$};
\node (hv) at (2,-1) {$\hat v$};
\node (bv) at (3.5,-1) {$\bar v$};
\node (hu) at (5,1) {$\hat u$};
\node (t') at (6,0) {$t'$};
\node (t) at (7.5,0) {$t$};
\draw [->, decorate, decoration={snake}, red] (s') to (hv);
\draw [->, decorate, decoration={snake}, red] (bv) -- (bu);
\draw [->, decorate, decoration={snake}, red] (hu) -- (t');
\draw [->, decorate, decoration={snake}, blue] (s) -- (s');
\draw [->, decorate, decoration={snake}, blue] (s') -- (bu);
\draw [->, decorate, decoration={snake}, blue] (hv) -- (bv);
\draw [->, decorate, decoration={snake}, blue] (bu) -- (hu);
\draw [->, decorate, decoration={snake}, blue] (bv) -- (t');
\draw [->, decorate, decoration={snake}, blue] (t') -- (t);
\end{tikzpicture}
\]
This is an embedding of $W$ in $G$; indeed:
\begin{enumerate}
\item $s' \PPleads \hat v$ intersects
\begin{itemize}
\item $s \Pleads s'$ only in $s'$, otherwise $MIS(G)$ would contain a cycle;
\item $s' \Pleads \hat u$ only in $s'$, because $s'$ is the last intersection;
\item $\hat v \Pleads t$ only in $\hat v$, otherwise $MIS(G)$ would contain a cycle.
\end{itemize}

\item $\bar v \PPleads \bar u$ intersects
\begin{itemize}
\item never $s \Pleads s'$, otherwise $MIS(G)$ would contain a cycle;
\item $s' \Pleads \bar u$ only in $\bar u$, because $\bar u$ is the first intersection;
\item $\hat v \Pleads \bar v$ only in $\bar v$, otherwise $MIS(G)$ would contain a cycle;
\item $\bar u \Pleads \hat u$ only in $\bar u$, otherwise $MIS(G)$ would contain a cycle;
\item $\bar v \Pleads t'$ only in $t'$, because $\bar v$ is the last intersection;
\item never $t' \Pleads t$, otherwise $MIS(G)$ would contain a cycle.
\end{itemize}

\item $\hat u \PPleads t'$ intersects
\begin{itemize}
\item $s \Pleads \hat u$ only in $\hat u$, otherwise $MIS(G)$ would contain a cycle;
\item $\hat v \Pleads t'$ only in $t'$, because $t'$ is the first intersection;
\item $t' \Pleads t$ only in $t'$, otherwise $MIS(G)$ would contain a cycle.
\qedhere
\end{itemize}
\end{enumerate}
\end{enumerate}
\end{proof}

\begin{proof}[\bf Proof of Corollary~\ref{cor:newMIS}]
Clearly, every simple $st$-path in $G$ remains simple in $G \cup \{e\}$; hence, all $MIS(G)$ belongs to $MIS(G\cup\{e\})$. Moreover,
by Lemma~\ref{Prop:RE}, only the edges and vertices in $RCC(e)$ that become irredundant after the addition of $e$ are in $MIS(G\cup\{e\})$.
\end{proof}

\begin{proof}[\bf Proof of Lemma~\ref{irrvuln}]
 ($\Leftarrow$) Because the addition of $e$ is impactful with entry-exit pair $(\en,\ex)$, we have that $\en<\ex$ in $MIS(G)$; so, in $G$ there are  two disjoint simple paths $s\rightsquigarrow \en$ and $\ex \rightsquigarrow t$ that are also disjoint from $RCC(e)\setminus\{\en,\ex\}$.
Thus, the embedding of $W$ into $RCC(e)$ can be used to embed $W$ into $G \cup \{e\}$.
 
 ($\Rightarrow$) Since $G$ is not vulnerable but $G\cup \{e\}$ is, it means that the vulnerability comes from some edges that were redundant but now are irredundant, or from $e$. Since the addition of $e$ is impactful, by Proposition~\ref{irredirred} $e$ is irredundant and $RCC(e)$ is a new parallel component w.r.t. $MIS(G)$. By Proposition~\ref{Prop:RE}, all the redundant edges in $G$ not in $RCC(e)$ remain redundant in $G\cup \{e\}$. Thus, the only possibility is that $RCC(e)$, seen as an $\en\ex$-net, is vulnerable itself.
\end{proof}


\section{Proofs for Section 4}

\begin{proof}[\bf Proof of Proposition~\ref{prop:Tu}]
The if part is trivial by transitivity of $\leq$ (see Proposition~\ref{prop:leq}). 
So, by contradiction, let us assume that $A\leq B$ but there is no such node $w$. 
The interesting case is when $|A|\geq 2$ and $|B|\geq 2$: indeed, if $A=\{a\}$ (resp., $B=\{b\}$), 
the node $a$ (resp., $b$) can play the role of $w$ in the statement.
Hence, let $a, a'\in A$ and $b, b'\in B$; again, $a$ and $a'$, as well as $b$ and $b'$, must be incomparable (w.r.t. $\leq$), otherwise one of them can play the role of $w$.

Now consider the simple $st$-paths $s \leadsto a\leadsto b \leadsto t$, $s \leadsto a'\leadsto b' \leadsto t$, and $a \leadsto b'$;
let $s'$ be the last vertex common to $s \leadsto a$ and $s \leadsto a'$ (at worst, $s$), and $t'$ be the first vertex common to $b \leadsto t$ and $b' \leadsto t$ (at worst, $t$). Furthermore, let $i_1$ be the last intersection between $a\leadsto b$ and $a\leadsto b'$ (at worst, $a$) and
$i_2$ be the first successive intersection between $a\leadsto b'$ and $a'\leadsto b'$ (at worst, $b'$).
Pictorially, we can represent the situation as follows:
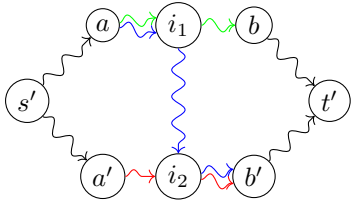
\[
\begin{tikzpicture}[every node/.style={circle,draw, minimum size=4mm, inner sep=2pt}]
\node (s') at (0,0) {$s'$};
\node (a) at (1,1) {$a$};
\node (i1) at (2,1) {$i_1$};
\node (b) at (3,1) {$b$};
\node (a') at (1,-1) {$a'$};
\node (i2) at (2,-1) {$i_2$};
\node (b') at (3,-1) {$b'$};
\node (t') at (4,0) {$t'$};
\draw [->, decorate, decoration={snake}] (s') -- (a);
\draw [->, decorate, decoration={snake}, green] ([yshift=10pt]a) -- ([yshift=10pt]i1);
\draw [->, decorate, decoration={snake}, green] (i1) -- (b);
\draw [->, decorate, decoration={snake}] (b) -- (t');
\draw [->, decorate, decoration={snake}, blue] ([yshift=-10pt]a) -- ([yshift=-10pt]i1);
\draw [->, decorate, decoration={snake}, blue] (i1) -- (i2);
\draw [->, decorate, decoration={snake}, blue] ([yshift=10pt]i2) -- ([yshift=10pt]b');
\draw [->, decorate, decoration={snake}] (s') -- (a');
\draw [->, decorate, decoration={snake}, red] (a') -- (i2);
\draw [->, decorate, decoration={snake}, red] ([yshift=-10pt]i2) -- ([yshift=-10pt]b');
\draw [->, decorate, decoration={snake}] (b') -- (t');
\end{tikzpicture}
\]
If we prove that all these paths are disjoint, we provide an embedding of $W$ into $G$, in contradiction with the hypothesis that $G$ is series-parallel.
First, we want to prove that intersection between the upper and the lower paths is $\{s',t'\}$; to this aim, we observe that:
\begin{itemize}
\item $s' \leadsto a\ \cap\ s' \leadsto a' = \{s'\}$ by construction;
\item $s' \leadsto a\ \cap\ a' \leadsto t' = \emptyset$ otherwise $a' < a$;
\item $a \leadsto b\ \cap\ s' \leadsto a' = \emptyset$ otherwise $a < a'$;
\item $a \leadsto b\ \cap\ a' \leadsto b' = \emptyset$ otherwise there exists a $w$ between $\{a,a'\}$ and $\{b,b'\}$;
\item $a \leadsto b\ \cap\ b' \leadsto t' = \emptyset$ otherwise $b' < b$;
\item $b \leadsto t'\ \cap\ s' \leadsto b' = \emptyset$ otherwise $b < b'$;
\item $b \leadsto t'\ \cap\ b' \leadsto t' = \{t'\}$ by construction.
\end{itemize}
Finally, let us focus on $i_1 \leadsto i_2$:
\begin{itemize}
\item $i_1 \leadsto i_2\ \cap\ s' \leadsto i_1 = \{i_1\}$ otherwise $G$ would have a cycle;
\item $i_1 \leadsto i_2\ \cap\ i_1 \leadsto b = \{i_1\}$ by construction and 
$i_1 \leadsto i_2\ \cap\ b \leadsto t = \emptyset$ otherwise $b < b'$;
\item $i_1 \leadsto i_2\ \cap\ s' \leadsto a' = \emptyset$ otherwise $a < a'$ and 
$i_1 \leadsto i_2\ \cap\ a' \leadsto i_2 = \{i_2\}$ by construction;
\item $i_1 \leadsto i_2\ \cap\ i_2 \leadsto t' = \{i_2\}$ otherwise $G$ would have a cycle.
\qedhere
\end{itemize}
\end{proof}

\begin{proof}[\bf Proof of Lemma~\ref{prop_cost-R-I-V}]
For the impactless case, we observe that $X \leq N$ whenever any of the two sets is empty (lines \ref{algS:check-empty}--\ref{algS:empty}) or whenever there exists a $w$ such that $X \leq w \leq N$ (by Proposition \ref{prop:Tu}). For the latter case, notice that,
after running $MinFoll(mis,X)$ in line \ref{algS:Tleqw} (and this costs $O(m)$ -- see Lemma~\ref{minFollCorrect}), we have a node $w$ that is minimal (w.r.t. $\leq$) such that $X\leq w$ (see Lemma~\ref{minFollCorrect} later on); thus, it suffices to check whether $w \leq N$ (i.e., whether the largest subset $N'$ of $N$ whose nodes are all reachable from $w$ is all $N$ -- that is, the two sets have the same cardinality, being one a subset of the other); this is done in lines \ref{algS:impless1}--\ref{algS:impless2}. Notice that calculating $N'$ is done through any visit from $w$ to $N$ in $G$ and so it costs $O(m)$.

If $|N'| < |N|-1$, this means that there exist two entries that are not proceeded by the exits; so, there exists at least one entry, different from the entry of the entry-exit pair for $e$, that is not proceeded by all exits. Hence, the addition of $e$ cannot be impactless nor impactfull; so, it must be destructive (see lines \ref{algS:destr1}--\ref{algS:destr2}).

Then, we calculate $X'$, made up by all the nodes of $X$ that proceed the only entry $\en$ that is not proceeded by $w$; this can be done by a visit in $G^T$ from $\en$ to $X$ (i.e., a backwards reachability) and costs $O(m)$. If $X'$ coincides with $X$, then no exit is proceeded by $\en$ and so no entry-exit pair for $e$ can exist; if there is more than one exit that does not proceed $\en$, then again the addition of $e$ cannot be impactful. In both cases, Algorithm~\ref{alg:Salvo} correctly says that the addition of $e$ is destructive (see lines \ref{algS:destr3}--\ref{algS:destr4}).

So, after line \ref{algS:destr4}, we are left with just one pair $(\en,\ex)$ such that $\ex \not\leq \en$; this can potentially be the entry-exit pair for the impactful addition of $e$. To verify this, we have to check that
\begin{itemize}
\item {\em $\en < \ex$:} this is simply a reachability from $\en$ to $\ex$ in $MIS(G)$ and costs $O(m)$.
\item {\em $\en$ and $\ex$ are linkable:} this can be done by adding $(\en, \ex)$ to $MIS(G)$ and by checking whether the resulting graph is still series-parallel; by Remark~\ref{rem:linkable}, the cost is $O(m)$.
\item {\em all other pairs made up by an entry $\en'$ and an exit $\ex'$ are such that $\ex' \leq \en'$:}\linebreak
this is implied by checking that $X'\ (= X \setminus \{\ex\}) \leq N$ and $X \leq N'\ (= N \setminus \{\en\})$.
These last checks can be again done by relying on Proposition~\ref{prop:Tu}: for example, for checking $X' \leq N$, we run $MinFoll$ 
to find a node $w'$ such that $X' \leq w'$ and then we do a visit in $G$ from $w'$ to $N$ to check if $w'\leq N$.
Again, this costs $O(m)$, since $MinFoll$ costs $O(m)$ (see Remark~\ref{minFollCorrect}).
\end{itemize}
If all these checks succeed, Algorithm~\ref{alg:Salvo} correctly says that the addition of $e$ is impactful, otherwise it correctly says that it is destructive.
\end{proof}

\begin{proof}[\bf Proof of Proposition~\ref{prop:prec}]
By induction on the topological order used in line 17 of Algorithm~\ref{alg:minFoll}.
The base case is for $s$ and is immediate: by line 15, $prec_A[s]$ is always defined and, by lines 1--4, it is 1 if $s \in A$ and 0 otherwise.
For the inductive case, if $v = t$ then the case is trivial as well, since $t$ covers all vertices of $G$, and so in particular all $A$ (see line 16).
Otherwise, by line 11, $v$ has a $ref$, say it is $u$. By Lemma~\ref{lem:npA}, we have that $np_A[v] = |\{a \in A \,:\, u < a \leq v \text{ in } G\}|$.
Moreover, since $u$ must appear before $v$ in the topological order, $prec_A[u]$ is defined and, by induction, $prec_A[u]=|\{a \in A \,:\, a \leq u  \text{ in } G\}|$. By looking at line 19, we can conclude, since $|\{a \in A \,:\, a \leq v \text{ in } G\}| = |\{a \in A \,:\, a \leq u \text{ in } G\}| + |\{a \in A \,:\, u < a \leq v \text{ in } G\}|$.
\end{proof}

\begin{proof}[\bf Proof of Lemma~\ref{minFollCorrect}]
By lines 17--21 of Algorithm~\ref{alg:minFoll}, if $w$ is returned, then its $prec_A$ has been set to $|A|$; 
by Proposition~\ref{prop:prec}, $A \leq w$. Moreover, if $MinFoll$ returned $w$, this means that $w$ is the first one in the chosen topological order with this property; hence, if $a \leq w'$, it cannot be that $w' \leq w$, otherwise $w'$ must appear in the topological order before $w$ and so $MinFoll$ would not have returned $w$. So, it can either be that $w \leq w'$ (as desired), or that $w$ and $w'$ are unrelated w.r.t. ‘$\leq$'. 
But this letter case is not possible since otherwise we would have that $A \leq \{w,w'\}$ and so, by Proposition~\ref{prop:Tu}, there would exist another $w'' (\neq w,w')$ such that $A \leq w'' \leq \{w,w'\}$; but then $MinFoll$ would have returned $w''$.

Being a (minor) variation of the algorithm in \cite{VTL82}, the complexity of Algorithm \ref{alg:minFoll} is: 
(1) $O(n)$ for lines 1--4;
(2) $O(n+m)$ for lines 5--14;
(3) $O(n)$ for lines 15--21.
So, overall the complexity of the algorithm is $O(m)$, being the input graph $st$-series-parallel.
\end{proof}

\subsection{Proof of Lemma~\ref{lem:npA}}
\label{app:proof}

To prove the Lemma, we found it easier to enrich Algorithm~\ref{alg:minFoll} with two additional data structures, that will enable a simpler formulation and proof of the desired result. The result is Algorithm~\ref{alg:minFollInv}, where we also introduce
 \begin{itemize}
     \item for every pair $u,v \in V_G$, a set of vertices $precE[u, v]$ that contains all the vertices of $A$ contained in the $uv$-subgraph of $G$, $u$ excluded; and
     \item for every vertex $v$, a set of vertices $precV[v]$ that contains all the vertices of $A$ covered by $v$ in $G$ contained in the subgraph of $G$ from any predecessor of $v$ in $\hat G$ (the predecessors excluded), that is, the union of all the $precE[u, v]$.
 \end{itemize}
 The formal definition of such sets is stipulated by the formulae $\varphi_1$ and $\varphi_2$ below. 
We remark that the algorithm we shall use in practice is Algorithm~\ref{alg:minFoll} (where we just keep the number of the predecessors and not which are them); this keeps the computational cost in $O(m+n)$. The blue parts in Algorithm~\ref{alg:minFollInv} have only been added to formulate the following Lemma~\ref{lem:invariants}, that proves some loop invariants expressed in terms of $precV$ and $precE$.

\begin{algorithm}
    \caption{$MinFollInv(G, A)$ \\ Like $MinFoll$ but with extra data structures used to formulate and prove invariants (REMARK: the new parts w.r.t. Algorithm~\ref{alg:minFoll} are depicted in blue)} 
    \label{alg:minFollInv}
    \begin{algorithmic}[1]
        \ForAll{$v\in V$}
    	\If{$v \in A$}
{	\State $np_A[v] = 1$
    {\color{blue}
 	    \State $precV[v]=\{v\}$}
            {\color{blue}            \ForAll{$u\in V$} 
            \State $precE[u, v] = \{v\}$
            \EndFor}}
	\Else
{	    \State$np_A[v] = 0$
	    {\color{blue}\State $precV[v]=\emptyset$}
	   {\color{blue}\ForAll{$u\in V$} 
            \State $precE[u,v] = \emptyset$
            \EndFor}}
	\EndIf
        \EndFor
	\State Let $\hat{G} = (\hat{V},\hat{E})$, with $\hat{V} = V$ and $\hat{E} = E$ 
        \State Remove from $\hat{E}$ all parallel edges
        \While{$\hat{G} \neq s \rightarrow t$} 
        \State Let $v \in \hat{V}$ be such that $in\_degree_{\hat G}(v) = out\_degree_{\hat G}(v) = 1$
        \State Let $(u , v)$ and $(v,w)$ belong to $\hat{E}$
        \State $np_A[w]$ += $np_A[v]$  \label{it1}
        \State $ref[v] = u$
        \State \textcolor{blue}{$precE[u, w]=precE[u, v]\cup precE[v, w]\cup precE[u, w]$}\label{it2}
        \State \textcolor{blue}{$precV[w]=precV[w]\cup precV[v]$} \label{it3}
        \State Cancel $v$ from $\hat{V}$ and $(u , v)$ and $(v,w)$ from $\hat{E}$
        \If{$(u,w) \not\in \hat{E}$}
        \State add $(u,w)$ to $\hat{E}$
        \EndIf
        \EndWhile
        \State $prec_A[s] = np_A[s]$
        \State $prec_A[t] = |A|$
        \ForAll{$v \in V$ touched during a topological visit of $G$}
        \If{$v \neq s,t$}
        \State $prec_A[v] = np_A[v] + prec_A[ref[v]]$
        \EndIf
        \If{$prec_A[v] = |A|$}
        \State \Return $v$
        \EndIf
        \EndFor
    \end{algorithmic}
    \end{algorithm}

\newcommand{\pathPlus}{\overset{+}{\rightsquigarrow}}
\newcommand{\innNodes}[1]{\overset{\circ}{#1}}

Notationally, we let $G=(V,E)$ be the starting graph, $\hat{G}=(\hat{V},\hat{E})$ be the graph at the beginning of a generic iteration, and $\hat{G}'=(\hat{V}',\hat{E}')$ be the graph at the end of that generic iteration (if in that iteration we cancel vertex $v$, then $\hat{V}=\hat{V}'\cup\{v\}$). We use $precE'$, $precV'$ and $np_A'$ to denote the values of $precE$, $precV$ and $np_A$ when they refer to $\hat{G}'$, whereas the ‘‘non-primed" versions refer to $\hat G$.
Moreover, $a \pathPlus b$ denotes a non-empty path from $a$ to $b$. Finally, given a path $P$, we denote with $\innNodes P$ the internal nodes of $P$ (extremes excluded). 

Armed with these notations, we define the following logical formulae (that we shall prove to be invariants for the loop in lines 14--23 of Algorithm~\ref{alg:minFollInv}):
\[
\begin{array}{rcl}
\varphi_1(G,\hat{G}) & \triangleq & \forall (b,c) \in \hat E\, .\ precE[b , c] = 
\left \{ \!\!
\begin{array}{ll}
a\in A\ |\ \exists P \text{ in } G\ . & \!\!\!\!\! P =b \pathPlus a \rightsquigarrow c \\
& \!\!\!\!\!\wedge\ \innNodes P \cap \hat V = \emptyset
\end{array}
\!\!\right\}
\vspace*{.2cm}
\\
\varphi_2(G,\hat{G}) & \triangleq & \forall b \in \hat V\, .\ precV[b] = \underset{(x,b) \in \hat{E}}{\bigcup}precE[x,b]
\vspace*{.2cm}
\\
\varphi_3(G,\hat{G}) & \triangleq & \forall b,c,d,x \in \hat V\, .
        \left(\!\!
        \begin{array}{l} 
        (c\neq d) \Rightarrow (precE[b,c] \cap precE[x,d] = \emptyset)\ \wedge
        \\
        (x\neq b \wedge c \not\in A) \Rightarrow (precE[b,c] \cap precE[x,c] = \emptyset)\ \wedge
        \\
        (x\neq b \wedge c \in A) \Rightarrow (precE[b,c] \cap precE[x,c] = \{c\})
        \end{array}
        \!\!\right)
\vspace*{.2cm}
\\
\varphi_4(G,\hat{G}) & \triangleq & \forall b,c \in \hat V\, .\ (b\neq c) \Rightarrow (precV[b] \cap precV[c] = \emptyset)
\vspace*{.2cm}
\\
\varphi_5(G,\hat{G}) & \triangleq & \forall b \in \hat V\, .\ np_A[b] = |precV[b]|
\vspace*{.2cm}
\\
\varphi(G,\hat{G}) & \triangleq & \varphi_1(G,\hat{G})\land\varphi_2(G,\hat{G})\land\varphi_3(G,\hat{G})\land\varphi_4(G,\hat{G})\land\varphi_5(G,\hat{G})
\end{array}
\]
Finally, let $lg(\hat{G}) \triangleq \hat{G} \neq s \rightarrow t$ be the loop guard of the iteration of the {\bf while} starting with $\hat G$.
 
\begin{lemma}[Loop invariants]
\label{lem:invariants}
    If $\varphi(G,G)$ then $(\varphi(G,\hat{G})\land lg(\hat{G}))\Rightarrow \varphi(G,\hat{G}')$.
\end{lemma}
 \begin{proof}
 We first note that $\varphi(G,G)$ is trivial, because of lines 1--11.
 We now prove $(\varphi(G,\hat{G})\land lg(\hat{G}))\Rightarrow \varphi_i(G,\hat{G}')$, for $i=1,...,5$.
 Since $lg(\hat{G}) == TRUE$, it must exist a $v\in \hat V$ such that $(u , v)$ and $(v,w)$ are the only edges $\hat{E}$ touching it.
\begin{enumerate}
    \item If $(b, c)\neq (u, w)$, then $precE'[b, c]=precE[b, c]$ and the thesis follows because of $\varphi_1(G,\hat{G})$.
      Let now be $(b, c)=(u, w)$; we want to prove that 
        \begin{equation}\label{eq:phi1}
        precE'[u, w]=\{a\in A\ |\ \exists P \text{ in } G\, :\, P=u \pathPlus a \rightsquigarrow w \, \wedge\, \innNodes P \cap \hat V' =\emptyset\}
        \quad
        \end{equation}
        By line~\ref{it2}, we know that 
        $$
        precE'[u, w]=precE[u, v]\cup precE[v, w]\cup precE[u, w]
        $$ 
        and, by $\varphi_1(G,\hat{G})$, we know that
        \begin{itemize}
            \item $precE[u, v]=\{a\in A\ |\ \exists P$ in $G: P=u\pathPlus a\rightsquigarrow v\ \wedge\ \innNodes P \cap \hat V=\emptyset\}$,
            \item $precE[v, w]=\{a\in A\ |\ \exists P$ in $G: P=v\pathPlus a\rightsquigarrow w\ \wedge\ \innNodes P \cap \hat V=\emptyset\}$,
            \item $precE[u, w]=\{a\in A\ |\ \exists P$ in $G: P=u\pathPlus a\rightsquigarrow w\ \wedge\ \innNodes P \cap \hat V=\emptyset\}$.
        \end{itemize}
        By $\varphi_3(G,\hat{G})$, we know that
        \begin{itemize}
            \item $precE[u, v]\cap precE[v, w]=\emptyset$,
            \item $precE[u, v]\cap precE[u, w]=\emptyset$,
            \item $precE[v, w]\cap precE[u, w]$ is $\{w\}$, if $w \in A$, and is $\emptyset$, otherwise.
        \end{itemize}
        We now prove \eqref{eq:phi1} by double inclusion:
        \begin{itemize}
            \item [($\subseteq$)] By contradiction, $\exists a\in A$ such that $a \in precE[u, w]$ but $a\notin precE[u, v]$, $a\notin precE[v, w]$, and  $a\notin precE[u, w]$. 
            Since $a\in precE'[u, w]$ and $a\notin precE[u, w]$, and since $\hat{V}'=\hat{V}\setminus\{v\}$, the only possibility is that all $P=u\pathPlus a \rightsquigarrow w$ in $G$ are such that $\innNodes P\cap \hat{V}=\{v\}$. 
            Now consider a generic such $P$; then, within $P$, $v$ can occur either before $a$ or after it.
            In the first case, we would have that $P=u \pathPlus v\rightsquigarrow a\rightsquigarrow w$, 
            but this would imply that $a\in precE[v, w]$.
            In the second case, we would have that $P=u \pathPlus a \rightsquigarrow v \rightsquigarrow w$, 
            but this would imply that $a\in precE[u, v]$.
            
            \item [($\supseteq$)] By contradiction, $\exists a\in A$ such that $a \in (precE[u, v]\cup precE[v, w]\cup precE[u, w])$ but $a\notin precE'[u, w]$. 
            If $a=w$, then by construction $w\in precE'[u, w]$: indeed, $precE[u, w] = \{w\}$ in line 6 and the algorithm never removes elements from $precE$.
            Hence, we can assume that the three $prec_E(\cdot)$ are pair-wise disjoint and we can consider the three cases is isolation:
            \begin{itemize}
            \item $a\in precE[u, w]$: this would imply that $a\in precE'[u, w]$, since $precE$ cannot decrease.
            \item $a\in precE[u, v]$: let $P_1$ be a path in $G$ such that $P_1= u \pathPlus a \leadsto v$ and $\innNodes P_1\cap \hat{V}=\emptyset$. 
            By construction, there exists in $G$ a $P_2=v\rightsquigarrow w$, and this must be disjoint from $P_1$ (otherwise, there would be a cycle in $G$). 
            Let us consider the path in $G$ $u\overset{P_1}\rightsquigarrow a \overset{P_1}\rightsquigarrow v \overset{P_2}\rightsquigarrow w$; then, there must exist an $x\in P_2\cap \hat{V}$ such that $u\overset{P_1}\rightsquigarrow a \overset{P_1}\rightsquigarrow v\overset{P_2} \rightsquigarrow x \overset{P_2}\rightsquigarrow w$. 
            This is not possible, otherwise $out\_degree_{\hat G}(v)$ would not be 1.
            
                \item $a\in precE[v, w]$: let $P_2$ be a path in $G$ such that $P_2= v \pathPlus a \leadsto w$ and $\innNodes P_2\cap \hat{V}=\emptyset$. 
                By construction, there exists in $G$ a $P_1=u\rightsquigarrow v$, and this must be disjoint from $P_2$ (otherwise, there would be a cycle in $G$).
                Let us consider $u\overset{P_1}\rightsquigarrow v\overset{P_2}\rightsquigarrow a  \overset{P_2}\rightsquigarrow w$; then, there must exist an $x\in P_2 \cap \hat{V}$ such that $u \overset{P_1}\rightsquigarrow x\overset{P_1}\rightsquigarrow v \overset{P_2}\rightsquigarrow a \overset{P_2}\rightsquigarrow w$.
                This is not possible, otherwise $in\_degree_{\hat G}(v)$ would not be 1.
            \end{itemize}
            \end{itemize}
    
    \item If $b\neq w$, then $precV'[b]=precV[b]$ and the thesis follows by $\varphi_2(G,\hat G)$.
    Let $b=w$, we want to prove that
    \begin{equation}\label{eq:phi2}
    precV'[w]=\underset{(x, w) \text{ in } \hat{E}'}{\bigcup}precE'[x, w]
    \end{equation}
    First, let us observe that, for all $x\in \hat{V}'$ such that $x\neq u$, by construction $precE'[x, w]=precE[x, w]$.
    Then, since $(v,w) \not\in \hat{E}'$, we have that 
    \begin{equation}\label{eq:phi2-app}
    \underset{(x, w) \in \hat{E}' : x\ne u}{\bigcup}precE'[x, w]=
    \underset{(x, w) \in \hat{E} : x\not\in \{u,v\}}{\bigcup}precE[x, w]
    \end{equation}
    So, we can prove \eqref{eq:phi2} as follows:
    $$
    \begin{array}{rcll}
    precV'[w] & = & precV[w]\cup precV[v] & \text{(by line~\ref{it3})}
    \vspace*{.2cm}
    \\
    & = & 
    \underset{(x, w) \in \hat{E}}{\bigcup}precE[x, w]\ \ \cup\ \ 
    \underset{(x, v) \in \hat{E}}{\bigcup}precE[x, v]
    & \text{(by $\varphi_2(G,\hat G)$)}
    \vspace*{.2cm}
    \\ 
    & = & 
    \underset{(x, w) \in \hat{E}}{\bigcup}precE[x, w]\ \ \cup\ \ precE[u, v]
    & \text{($in\_degree_{\hat G}(v)=1$)}
    \vspace*{.2cm}
    \\ 
    & = & \underset{(x, w) \in \hat{E}:x\not\in \{u, v\}}{\bigcup}precE[x, w]\ \ \cup\ precE[u, w]
    \vspace*{-.35cm}
    \\
    && \hspace*{4.35cm}\cup\ \,precE[v, w]
    \\
    && \hspace*{4.35cm}\cup\ \,precE[u, v]
    \vspace*{.2cm}
    \\ 
    & = & \underset{(x, w) \in \hat{E}:x\not\in \{u, v\}}{\bigcup}precE[x, w[\ \ \cup\ precE'[u, w]
    & \text{(by line~\ref{it2})}
    \vspace*{.2cm}
    \\
    & = & \underset{(x, w) \in \hat{E}':x\ne u}{\bigcup}precE'[x, w]\ \ \cup\ precE'[u, w]
    & \text{(by \eqref{eq:phi2-app})}
    \vspace*{.2cm}
    \\
    & = & \underset{(x, w) \in \hat{E}'}{\bigcup}precE'[x, w]
    \end{array}
    $$

    \item For all pairs $(x,y)$ different from $(u,w)$, by construction $precE'[x,y]=precE[x,y]$; so, for all 4-tuples $b,c,d,x \in \hat V'$ in which $(u,w)$ does not appear, the thesis follows by using $\varphi_3(G,\hat{G})$. W.l.o.g., let $(b, c)=(u, w)$; we want to prove that
        \begin{enumerate} 
        \item $precE'[u, w] \cap precE'[x, d] = \emptyset$, whenever $d\neq w$;
        \item $precE'[u, w] \cap precE'[x, w] = \{w\}$, whenever $x\neq u$ and $w \in A$;
        \item $precE'[u, w] \cap precE'[x, w] = \emptyset$, whenever $x\neq u$ and $w \not\in A$.
        \end{enumerate}
First, $precE'[u, w]=precE[u, w]\cup precE[u, v]\cup precE[v, w]$, by line~\ref{it2};
moreover, since $v \not\in \hat{V}'$, we have that $d \ne v$ and $x \ne v$.
Then, (a) follows from $\varphi_3(G,\hat{G})$, since $d \not\in \{v,w\}$.
For (b) and (c), still by $\varphi_3(G,\hat G)$, we have that
            \begin{itemize}
                \item if $x\ne u$, then $precE[u, w]\cap precE[x, w]$ is $\{w\}$, if $w \in A$, and is $\emptyset$, otherwise;
                \item if $x\ne u$, then $precE[u, v]\cap precE[x, w]=\emptyset$;
                \item if $x\ne v$, then $precE[v, w]\cap precE[x, w]$ is $\{w\}$, if $w \in A$, and is $\emptyset$, otherwise.
            \end{itemize}
So, we can conclude that, whenever $x \ne u$, we have that $precE'[u, w] \cap precE'[x, w]$ is $\{w\}$, if $w \in A$, and is $\emptyset$, otherwise.

\item 
For all $x\ne w$, by construction $precV'[x]=precV[x]$; so, if both
$b\ne w$ and $c\ne w$, then the thesis follows by $\varphi_4(G,\hat{G})$.
W.l.o.g., let $b=w$ and $c\ne w$; by line~\ref{it3}, $precV'[w]=precV[w]\cup precV[v]$. 
Because of $\varphi_4(G,\hat{G})$, we have that:
        \begin{itemize}
            \item $precV[w]\cap precV[c]=\emptyset$, since $c\ne w$;
            \item $precV[v]\cap precV[c]=\emptyset$, since $c\ne v$ (being $v \not\in \hat V'$ and $c \in \hat V'$).
        \end{itemize}
This allows us to conclude $precV'[b] \cap precV'[c] = \emptyset$, as required.
    
\item 
If $b\ne w$, then, by construction, $np_A'[b] = np_A[b]$ and $precV[b] = precV'[b]$; by $\varphi_5(G,\hat G)$, we have that
$np_A[b] = |precV[b]|$ and conclude.
If $b=w$, then
\[
    \begin{array}{rcll}
        np_A'[w] & = & np_A[w] + np_A[v] & \text{(by line~\ref{it1})}
        \\
         & = & |precV[w]| + |precV[v]| & \text{(by $\varphi_5(G,\hat G)$)}
         \\
         & = & |precV[w] \cup precV[v]| & \text{(by $\varphi_4(G,\hat G)$, being $w \ne v$)}
         \\
         & = & |precV'[w]| & \text{(by line~\ref{it3})}
         \vspace*{-.7cm}
    \end{array}
\]
\end{enumerate}     
 \end{proof}
 
\begin{proof}[\bf Proof of Lemma~\ref{lem:npA}]
If $ref[v] = u$, then at some iteration of the {\bf while} in line 14 of Algorithm~\ref{alg:minFollInv} it happened that $v$ was selected for elimination, and this implies that, in that iteration, $in\_degree_{\hat G}(v)$ was 1 and $u$ was the only vertex such that $(u,v) \in \hat E$. Thus, by Lemma~\ref{lem:invariants}, $\varphi(G,\hat G)$ holds and so we have that
\[
\begin{array}{llll}
np_A[v] & = & |precV[v]| & \text{(by $\varphi_5(G,\hat G)$)}
\vspace*{.2cm}
\\
& = & \left | \underset{(x,v) \in \hat{E}}{\bigcup}precE[x,v] \right |  & \text{(by $\varphi_2(G,\hat G)$)}
\vspace*{.2cm}
\\
& = &  | precE[u,v] |  & \text{($in\_degree_{\hat G}(v)=1$)}
\vspace*{.2cm}
\\
& = & \left |
\left \{ 
\begin{array}{ll}
a\in A\ |\ \exists P \text{ in } G\ . & \!\!\!\!\! P = u \pathPlus a \rightsquigarrow v \\
& \!\!\!\!\!\wedge\ \innNodes P \cap \hat V = \emptyset
\end{array}
\right\}
\right |  & \text{(by $\varphi_1(G,\hat G)$)}
\vspace*{.2cm}
\\
& = & |\{a \in A \,:\, u < a \leq v \text{ in } G\}| & (\star)
\end{array}
\]
To prove $(\star)$, it suffices to note that, since $(u,v)$ is the only edge of $\hat E$ entering into $v$,  there cannot exist in $G$ a path from $u$ to $v$ that touches another vertex of $\hat V$.
We easily conclude, by noting that: (1) after the iteration in which $v$ is deleted from $\hat G$, its $np_A$ does not change anymore; and (2) the handling of $np_A$ is the same in Algorithms~\ref{alg:minFollInv} and~\ref{alg:minFoll}.
\end{proof}

\end{document}